\documentclass[letterpaper,twocolumn,10pt,usenames,dvipsnames]{article}
\usepackage{usenix-2020-09}

\newcommand{\commentOut}[1]{}

\usepackage{xurl}

\usepackage{aliascnt}

\newcommand{\anon}[1]{#1}
\newcommand{\protname}{MERGE}

\newcommand{\BBInclusionCheck}{BB Inclusion Check}
\newcommand{\StickerBBUpload}{Sticker BB Upload}

\newcommand{\nodelaynostuff}{Electronic-only}
\newcommand{\Eonly}{E-only}
\newcommand{\nodelay}{Electronic\&Stuff}
\newcommand{\nostuff}{Electronic\&Drop}
\newcommand{\EandD}{E\&D}

\newcommand{\remoteVotingCenter}{Remote Voting Center}
\newcommand{\voteCollectionCenter}{Local Counting Center}
\newcommand{\localVotingCenter}{Local Counting Center}

\newcommand{\remoteCommittedCVR}{Remote committed CVR}
\newcommand{\BB}{BB \hookleftarrow}
\newcommand{\enc}{\mathit{E}}
\newcommand{\mix}{\mathit{Mix}}
\newcommand{\dec}{\mathit{D}}
\newcommand{\hash}{\mathit{H}}

\newcommand{\sig}{\mathit{sig}}
\newcommand{\adv}{\mathcal{A}}
\newcommand{\auth}{\mathit{auth}}

\newcommand{\numvoters}{n}
\newcommand{\pos}{\textsc{pos}}
\newcommand{\Type}{\textsc{Type}}
\newcommand{\ValidZKP}[1]{\textit{ZKP}_{\textsc{Valid}}(#1)}
\newcommand{\DecryptZKP}[2]{\textit{ZKP}_{\textsc{Dec}}(#1,#2)}
\newcommand{\MixZKP}[2]{\textit{ZKP}_{\textsc{Mix}}(#1,#2)}

\newcommand{\commentOutForPublication}[1]{}
\hypersetup{
    hidelinks,
    colorlinks,
    breaklinks=true,
    urlcolor=link-blue,
    citecolor=color1,
    linkcolor=color1,
    bookmarksopen=false,
    pdftitle={Title},
    pdfauthor={Author},
}

\usepackage[all]{nowidow}

\usepackage{lastpage}
\usepackage{placeins}
\usepackage{graphicx}
\usepackage{float}
\usepackage{xspace}
\usepackage{longtable}
\usepackage{tabularx}
\usepackage{color,colortbl}
\usepackage[square,sort,comma,numbers]{natbib}
\usepackage{enumerate}
\usepackage{enumitem}
\usepackage{amssymb}
\usepackage{amsthm}
\usepackage{amsmath}
\usepackage{amsfonts}
\usepackage{rotating}
\usepackage{caption}
\usepackage{subcaption}
\newtheorem{remark}{Remark}

\newaliascnt{coro}{theorem}
\newtheorem{coro}[coro]{Corollary}
\aliascntresetthe{coro}

\newaliascnt{definition}{theorem}
\newtheorem{definition}[definition]{Definition}
\aliascntresetthe{definition}

\newaliascnt{prop}{theorem}
\newtheorem{prop}[prop]{Proposition}
\aliascntresetthe{prop}

\newaliascnt{assumption}{theorem}
\newtheorem{assumption}[assumption]{Assumption}
\aliascntresetthe{assumption}

\newaliascnt{lemma}{theorem}

\aliascntresetthe{lemma}

\definecolor{link-blue}{RGB}{6,69,173}
\definecolor{dark-green}{RGB}{52,133,62}
\definecolor{light-blue}{RGB}{127,180,240}
\definecolor{dark-blue}{RGB}{72,120,224}
\definecolor{heading-grey}{RGB}{128,128,128}
\definecolor{heading2-grey}{RGB}{200,200,200}
\definecolor{Critical}{RGB}{192,0,0}
\definecolor{High}{RGB}{255,0,0}
\definecolor{Medium}{RGB}{255,192,0}
\definecolor{Low}{RGB}{255,255,0}
\definecolor{Informational}{RGB}{94,185,255}

\definecolor{color1}{RGB}{0,0,90} % Color of the article title and sections
\definecolor{color2}{RGB}{220,220,255} % Color of the boxes behind the abstract and headings

\definecolor{darkturquoise}{HTML}{0097B2}
\definecolor{highlighttextcolor}{RGB}{  17, 120, 100 }

\usepackage[most]{tcolorbox}    	% for COLORED BOXES (tikz and xcolor included)
\definecolor{main}{RGB}{0,0,90} % Color of the article title and sections
\definecolor{sub}{RGB}{220,220,255} % Color of the
\tcbset{%
    enhanced,
    colback=white,
    colframe=black,
    drop shadow,
    float,
    floatplacement=t,
    fonttitle=\bfseries
}

\newtcolorbox{boxD}[1]{
    enhanced,
    colback = white,
    colframe = main,
    drop shadow,
    title=#1
}

\newtcolorbox{boxDwhite}[1]{
    enhanced,
    colback = white,
    colframe = main,
    drop shadow,
    title=#1,
    float,
    floatplacement=!t
}

\newtcolorbox{boxDwhitePlaceHere}[1]{
    enhanced,
    colback = white,
    colframe = main,
    title=#1,
    float,
    floatplacement=!h
}

\newtcolorbox{boxHwhite}[1]{
    enhanced,
    colback = white,
    colframe = main,
    drop shadow,
    title=#1,
    float,
    floatplacement=!h
}

\title{MERGE: Matching Electronic Results with Genuine Evidence \\
\Large{for verifiable voting in person at remote locations}}
\author{\anon{Ben Adida, John Caron, Arash Mirzaei, Vanessa Teague\thanks{Authors are listed alphabetically.} \\
\texttt{\normalsize  ben@adida.net, jcaron1129@gmail.com \{arash.mirzaei, vanessa.teague\}@anu.edu.au} \\
CAC-Vote project\thanks{
This is a protocol proposal
 for a research project. Like any security protocol, it requires
substantial supporting assumptions and processes in order to be secure -
 there is no claim that any US jurisdiction has, or will have, the
necessary supporting legislation or processes to run it as specified.
    \indent This research was, in part, funded by the U.S. Government. The views and conclusions contained in this document are those of the authors and should not be interpreted as representing the official policies, either expressed or implied, of the U.S. Government.  Distribution Statement A - Approved for public release: distribution is unlimited.
}}}

\begin{document}
	\maketitle

\subsection*{Abstract}
Overseas military personnel often face significant challenges in participating in elections due to the slow pace of traditional mail systems, which can result in ballots missing crucial deadlines. While internet-based voting offers a faster alternative, it introduces serious risks to the integrity and privacy of the voting process. We introduce the \protname{} protocol to address these issues by combining the speed of electronic ballot delivery with the reliability of paper returns. This protocol allows voters to submit an electronic record of their vote quickly while simultaneously mailing a paper ballot for verification. The electronic record can be used for preliminary results, but the paper ballot is used in a Risk Limiting Audit (RLA) if received in time, ensuring the integrity of the election. This approach extends the time window for ballot arrival without undermining the security and accuracy of the vote count.

\section{Introduction}

    Slow mail is one of the main barriers to electoral participation for many overseas military personnel.
    Some jurisdictions send blank ballots by paper mail and require ballot return by the same channel---this
    can be very slow, often resulting in ballots that are not returned by the deadline. Other jurisdictions allow
    for paperless voting by Internet, including email or pdf upload---this is very fast, but introduces unacceptable
    risks to integrity and privacy. One compromise is electronic delivery and paper returns, in which voters
    download a blank ballot and mail back a printed paper vote, having filled it in either manually or electronically.
    This has similar integrity and privacy properties to
    vote-by-mail, and takes only half the time that a fully paper solution would take. Nevertheless, the return mail
    delivery is often still too slow.

    This paper describes the \protname{} protocol, which enhances the basic scheme of electronic delivery and paper returns     in a way that increases speed without significantly compromising privacy or integrity. The main idea is to send     an untrusted electronic record back quickly, which can be confirmed if the paper ballot arrives in time for the Risk Limiting Audit (RLA),
    or distrusted otherwise. The RLA is a post-election auditing process that uses statistical methods to ensure with high confidence that the reported election outcome is correct, by manually comparing a random sample of paper ballots against the electronic vote tallies. Since RLA  normally happens a few days to a few weeks after the election day, this extends the time available for the paper ballot to arrive.

    Assume a public bulletin board (BB) which is an un-erasable,  authenticated broadcast channel with memory. We will extensively use
    the assumption that people in different locations can rely on seeing the same data on the BB.
    \begin{itemize}
        \item In the polling place, the voter makes both a verifiable paper record, which is placed inside an envelope and mailed, and an untrusted
        electronic record, which is encrypted and posted to the BB.
        \item The voter is issued a digital signature, presented on a sticker, which serves as verifiable evidence of participation. Envelopes are considered acceptable only if they bear stickers with valid signatures upon receipt. This also allows officials at the \voteCollectionCenter\ to notify the voter when their properly-signed envelope has arrived.\footnote{This is expected to be part of the practical implementation. It does not, however, form an important part of the
    cryptographic protocol because it does not include evidence that the paper ballot matched the electronic
    one. The protocol could be extended, with some extra mixing, to notify each voter of whether their specific ballot matched. The current
    version simply produces a collective tally of how many of the received ballots matched.}
        \item The untrusted electronic record can be incorporated into preliminary results and used as the cast vote
        record in an RLA.
        \item If the paper ballot arrives in time for the RLA, it is used as the ballot, just like any other RLA.
        \item If no paper ballot arrives that corresponds to a given electronic record, we use the
        phantoms-to-zombies approach of Banuelos and Stark~\cite{banuelos2012limiting} to ensure that the risk limit is
        met even though the electronic record is not trusted.
    \end{itemize}

        Two assumptions are unavoidable: first, there must be an accurate count of people who legitimately  participated; second, the paper ballot  the voter places into the envelope must accurately     reflect their intentions. Our protocol also requires someone to check the digital signature on every envelope. These all need to be supported by human verification processes.

        The paper and electronic records could be produced by scanning a hand marked paper ballot, or
    printing a paper ballot from a ballot marking device (BMD).\footnote{Scientific opinion is divided on whether human
    verification of BMDs is sufficiently accurate to justify the assumption that the printout matches the voter's
    intention. See \cite{appel2020ballot} and \cite{kortum2021voter} for examples of opposing views. At a minimum, this
    requires careful design to
    ensure voters are motivated and encouraged to verify, and there is something they can do if a misprint
    occurs.} See \autoref{sec:threatModelAndTrustAssumptions} for a detailed discussion of the trust assumptions.

 Compared to
plain electronic delivery and paper returns, MERGE's main
advantage is that, rather than needing ballot papers to arrive in time to be scanned for preliminary results, a vote can be safely counted  if it arrives before the RLA. Of course, an electronic commitment to the vote must still be made before the normal close of polls. In practice, for most jurisdictions, this
allows a few more weeks for mail to arrive. Its main disadvantage is that---if a
large number of ballot papers do not arrive in time---the consequent large discrepancies may cause an RLA to fall back
to a manual recount, when it would not have if those votes had simply been excluded from the beginning. 

\protname{} achieves \emph{Software Independence}~\cite{rivest2008notion} based on a collective, probabilistic method of cast-as-intended verification designed to fit in to an existing Risk Limiting Audit. We assume
that the attacker can control any electronic components, but at least one of the devices chosen for each verification step is honest. See \autoref{subsec:security-goals}. %A second version has a similar process, but splits the vote generation between two computers so that neither knows both the vote and the voter's identity.

All the cryptographic components  are  implemented using Microsoft's ElectionGuard library, or some other cryptographic
library openly available to the general public.

The election occurs in two classes of locations, each of which might have several instances.

\begin{description}
	\item[\remoteVotingCenter s] are controlled polling places in remote areas such as overseas military bases.
	\item[\localVotingCenter s] are state- or county-based electoral authorities, where the ordinary votes from most citizens are counted. These will also receive paper ballots from \remoteVotingCenter s and process them.
\end{description}

There are likely to be many \remoteVotingCenter s and many \localVotingCenter s, but we assume only one of each in this paper for simplicity.

\subsection{Authentication using CAC cards}
Voters at \remoteVotingCenter s own a smart card called a Common Access Card (CAC) card. It provides a Public Key Infrastructure (PKI)
to enable secure authentication and digital signatures. Each CAC card includes a (private) signing key and a certificate linking its ID to the corresponding public key. Observers and officials at each \voteCollectionCenter\ know the CAC IDs and certificates of the voters to be included. This knowledge is obtained either through individual CAC voters, who must register their public keys with their local electoral authorities, or via a more centrally organized process. Only votes with a valid CAC signature are accepted via the electronic channel.

    Changes to CAC cards during the protocol, such as name changes, eligibility changes
    or key refreshes between the time the person votes and the time the votes are tallied, are out of
    scope for this paper but will need to be handled with specific policies in practice.

\subsection{\protname{} within the voting ecosystem}
We assume that most votes are not from \protname{}, but are cast locally near the \localVotingCenter. The usual audit process is conducted at the \localVotingCenter, with \protname{} ballots incorporated from \remoteVotingCenter{}s as needed.

The ballot manifest---available to observers at the \localVotingCenter---includes both the ordinary local ballots and all the ballots from \remoteVotingCenter s. So do the  preliminary
Cast Vote Records (CVRs).

The RLA proceeds exactly as it would if all the votes had been cast locally: observers see the ballot manifest and preliminary CVRs, then they watch a transparent process for seeding the PRNG that will be used to generate ballot samples, then they check the sequence of sampled ballots.
    For local ballots, CVRs are made in the usual way (e.g. by scanning); for \protname\ ballots, preliminary CVRs come from the BB in encrypted form.
    If the sampled ballot is local, officials retrieve it in the usual way and compute its effect on the audit statistics. If the ballot was cast remotely, it will require some extra work to locate the mailed ballot and compare it with its electronic counterpart. It is also possible that the ballot was delayed in the mail---we deal conservatively with this situation.
    Figure~\ref{fig:overview} shows how \protname{} fits into the voting and auditing process.

    \begin{center}
\begin{figure}
	\includegraphics[scale=0.3]{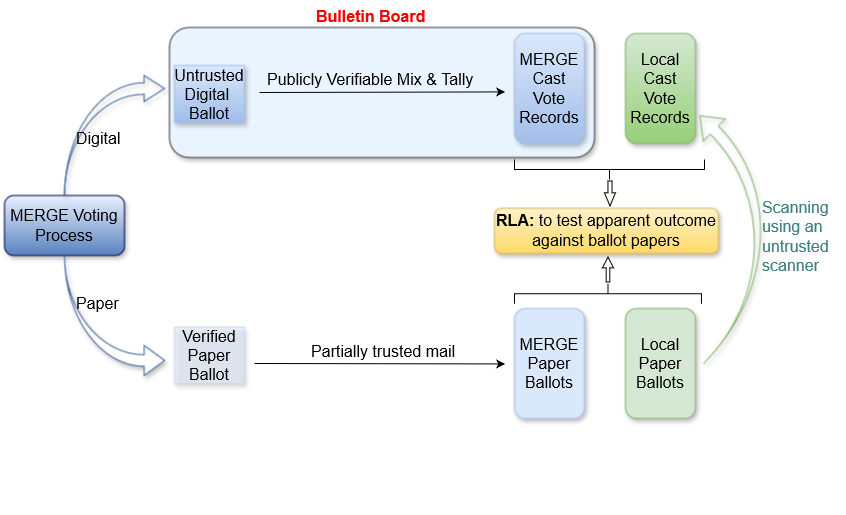}
	\caption{Overview of \protname{} in the voting ecosystem.}
	\label{fig:overview}
\end{figure}
\end{center}

\subsection{Security properties and contributions} \label{subsec:security-goals}
The \protname{} protocol satisfies the following properties:
\begin{itemize}
    \item \textbf{Authentication:} Only eligible voters can vote.
    \item \textbf{Privacy:} The protocol does not reveal more about an individual's vote than the published tallies do.\footnote{There are significant assumptions behind the privacy guarantees of \protname{}, which depend on which version is run. See \autoref{subsec:privacy}.} \protname{} does not achieve everlasting privacy.
    \item \textbf{Receipt Freeness:} It is infeasible for a voter to convince anyone of the value of their vote, even if they actively collude with the coercer, unless the coercer observes the serial numbers during the audit process.\footnote{\protname{} reveals who participated, so it is possible to coerce someone to refrain from participating altogether. See \autoref{subsec:security-receipt-freeness}.}
    \item \textbf{(Individual) paper-based cast-as-intended verification:} Each voter can see that their paper vote accurately reflects their intention and refuse to
    cast it if it does not.\footnote{Our cryptographic protocol is agnostic about whether this paper vote is filled in by hand by the voter or printed by a BMD. It obviously needs to be trustworthy in order for the RLA to be valid.}
    \item \textbf{(Public) tally verification:} anyone can verify that the  recorded electronic votes are correctly tallied.
    \item \textbf{(Collective, probabilistic) recorded-as-cast verification:} if the apparent election outcome (incorporating the electronic tally) does not accurately reflect the paper evidence (including both \protname{} and ordinary ballots), the RLA will fall back to a full manual count except with
    probability at most the risk limit.
    \end{itemize}

    The formalization and proof of the last property is a major contribution of this paper: we employ a game-based approach to demonstrate that \protname{} can interface with an existing RLA procedure, ensuring that the \emph{overall} risk limit is still met. While \cite{fuller2023adaptive} also uses a ``physical cryptography game'' framework to analyze RLAs in adversarial contexts, our work extends this approach to a complete e-voting protocol, addressing the complexities of processing both electronic and paper ballots prior to the RLA process. This property is different from (and neither weaker nor stronger than) end-to-end verifiability: although individuals cannot verify that \emph{their} vote is properly recorded on the BB, the system can instead guarantee a global limit on the risk that enough votes are misrecorded to alter the outcome.
    
    Guarantees achieved in this paper rely entirely on the trustworthiness of the paper trail, which is outside the control of the cryptographic protocol, and may fail for a variety of reasons, such as security problems in the mail channel and difficulties verifying printouts.  These need to be addressed by human procedures outside the \protname{} protocol.

    \subsection{Related work}
    There are many protocols for remote electronic voting, some with very strong security properties of various kinds. However, there is no single protocol that combines good privacy, receipt freeness (even against an attacker who can see only the bulletin board), public verifiability, and voter verification usable enough to constitute a good basis for trust. Since we work in a controlled setting, direct voter verification of a plaintext paper vote has substantial advantages, but the remote setting introduces practical challenges. This is a relatively under-studied area---most designs focus on either an uncontrolled Internet setting, or an in-person setting in the normal polling place.

    Many paperless Internet voting systems (e.g. \cite{adida2008helios}) provide end-to-end verifiability via the ``Benaloh Challenge'' method, in which the voter may choose to cast or challenge their encrypted vote. A challenge extracts a proof that can be checked on another device, and the process can be repeated. In theory this can give the voter a high probabilistic guarantee that the vote they cast matches their intention, but in practice the verification process is complex, which may lead ordinary voters to neglect it or get it wrong. The Estonian Internet voting system allows voters to challenge their real vote directly \cite{heiberg2014verifiable}.

    Remotegrity \cite{zagorski2013remotegrity} is a hybrid mail/internet extension of the Scantegrity II voting system \cite{chaum2008scantegrity}, designed for absentee voters. It is specifically tailored for pre-printed ballot systems, and uses code voting and scratch-off labels to enable voters to verify cast-as-intended and recorded-as-cast. Pretty Good Democracy (PGD) \cite{ryan2013pretty} is another code-sheet-based voting system where voters use codes to send their votes and receive only a single acknowledgment code as receipt. In Belenios-VS~\cite{cortier2019beleniosvs}, voters can receive their preprinted ciphertexts through mail and verify them using an independent device.

    Code-return voting systems, like the Norwegian \cite{gjosteen2012norwegian} and Swiss Internet voting systems, involve voters casting encrypted votes and receiving a confirmation code. They then cross-reference this code with a code sheet received via mail to verify their vote. In addition to individual verification of votes, these systems also enable server-side verification of the tally.

    The integrity of voting outcomes in code-based systems depends on the secrecy of the codes, a factor that voters cannot verify. This  obstructs achieving true end-to-end verifiability. Additionally, the secure printing and private transmission of code sheets pose significant practical challenges.

    Publishing vote ciphertexts that were built on a voter's machine introduces a potential coercion problem, because the voter could reveal the randomness that produced the ciphertext. This is a known issue in many systems including Helios. In Estonia, the authorities attempt to address the problem by allowing revoting, but this undermines verifiability~\cite{pereira2021individual}.
    Civitas \cite{clarkson2008civitas} prioritizes coercion resistance over ensuring that votes are cast as intended, as the voting device  is presumed to be trustworthy. Selene~\cite{ryan2016selene} posts the votes publicly in cleartext alongside private tracking numbers. Voters receive notification of their tracker after the vote/tracker pairs are posted, enabling them to select an alternative tracker if coerced.

    In summary, there are good reasons to take advantage of the simplicity of paper verification and the coercion-resistance of a voting computer that is not under the voter's control.

    ``Italian attacks'' may allow for coercion anyway: if there are many options on one ballot, then the existence of a particular pattern can be evidence of compliance with coercion. Some works therefore hide all information about individual ballots, including tallies~\cite{DBLP:conf/eurosp/KustersL00020, DBLP:conf/ccs/HuberKKL00R022}. \protname{} neither solves nor exacerbates this problem, assuming that mixed individual ballots are shared with auditors and observers---a higher level of protection seems hard to reconcile with an RLA.

    Many systems provide end-to-end verifiability for polling-place voting and combine it with plain paper verification.
    Prêt à Voter~\cite{ryan2009pret} uses preprinted auditable ciphertexts, and has been used in remote, controlled polling places~\cite{culnane2015vvote}. Voters  may audit preprinted blank ballots, and then vote on one that has not been opened. STAR-Vote~\cite{bell2013star} and wombat~\cite{wombat} adopt the ``cast-or-audit'' approach combined with plain paper, as does ElectionGuard~\cite{benaloh-electionguard}.  Scantegrity II~\cite{chaum2008scantegrity} combines a plain paper vote with short return codes for verification. Themis~\cite{DBLP:conf/ccs/BougonCCDDDGT22} combines a plain paper record with an innovative blinding-based verification process. All of the paper-based designs could incorporate an RLA. Most could be adapted to remote, controlled polling places, but none have detailed designs for dealing with ballots that are lost in the mail.

    ElectionGuard's specification~\cite{electionguard} also includes a vote-by-mail design, but this is intended for situations where preprinted ballots can be sent to voters and returned. It does not attempt to solve the timing issues associated with two-way mail.

    Devillez et al's ``Verifiable and Private vote-by-mail'' \cite{devillez2024verifiable} adds cryptographic verifiability to the electronic delivery and paper return strategy, much like \protname{}.  The main difference is that there is no pre-emptive electronic record sent, so no opportunity to take advantage of a timing difference between initial canvas and audit.

    %\subsubsection{Conclusion}
    %No easily-usable system with cast-as-intended verifiability exist. Following a different direction, our approach focuses on enabling cast-as-intended verification by prioritizing the alignment of paper ballots with voters' intentions, mirroring traditional paper-based voting practices. However, rather than waiting for mailed paper ballots to arrive, we promptly tally digital ballots and then upon the arrival of paper ballots employ the RLA method to ensure their consistency with the digital tally. Voters are empowered to verify all digital ballot processes and actively participate in the audit procedures.

\section{Technical Background and Notation} \label{sec:technicalBackground}

\subsection{Cryptographic components}
\paragraph{ElectionGuard}
ElectionGuard~\cite{benaloh-electionguard} is a suite of open-source tools developed by Microsoft to support end-to-end verifiable voting.
    It uses a variant exponential form of the ElGamal cryptosystem to encrypt all selections in a vote. This allows:
       \begin{itemize}
           \item homomorphic addition of encrypted values, enabling the decryption of only the overall election tally,
           \item re-randomization of  an encrypted value, which produces a ciphertext indistinguishable from a fresh one,
           \item threshold distributed decryption, which prevents any set of fewer than a threshold of talliers from colluding to decrypt individual votes.
           \end{itemize}

    \begin{remark}
    In this work, we use the elliptic curve exponential ElGamal encryption with threshold distributed decryption.
\end{remark}

    Our design uses the ElectionGuard toolkit and relies on the paper record rather than trying to gain trust in the electronic record, with some extra details to allow for the sorts of failures that can happen when votes are mailed.

 Homomorphic addition is written like multiplication on ciphertexts, either $\cdot$ for a pair or $\Pi$ for a product of multiple values.
    The encryption of vote vector $\mathbf{b}$ is denoted by $\mathbf{e}=\enc(\mathbf{b})$.

\paragraph{Non-interactive zero-knowledge (NIZK) proofs}
    The election transcript includes Noninteractive Zero Knowlege Proofs (NIZKPs) that allow for public verification that input votes are valid and that all decryptions are computed properly.

   A NIZK proof of validity for an encrypted vote vector $\mathbf{e}=\enc(\mathbf{b})$ is denoted by $\ValidZKP{\mathbf{b}}$. It proves:
\begin{itemize}
    \item the encryption for each option (e.g. ${e}_j$) encodes a valid value, usually either zero or one (i.e. ${b}_j \in \{0,1\}$), and
    \item the sum of all encrypted options within each contest falls within a specific range (determined by the voting rules).
\end{itemize}

We also use ElectionGuard's NIZK proofs to demonstrate that an encrypted value (e.g.  $e=\enc(b)$) is decrypted to a specific value (e.g. $b$), without revealing the decryption key. This proof  is denoted by $\DecryptZKP{b}{e}$.

    Of course, ZKPs cannot prove that the ciphertext accurately represents the voter's intention. For this, ElectionGuard uses a combination of challenge-based ``Benaloh challenge'' verification with plain paper records.

    More details about ElectionGuard are in Appendix~\ref{subsec:election-guard}.

    \paragraph{Private decryption}
Sometimes a plaintext, with a proof of proper decryption, will be supplied to some participants but not published on the BB. We call this \emph{private decryption.}

    \paragraph{Mixnet}  To maintain the confidentiality of each voter's ballot, a mixnet is used to mix it with others. It takes a set of encrypted ballots and outputs a shuffled set of re-encrypted ballots, making it challenging to link any specific encrypted ballot in its output to any of its input encrypted ballots. The mixnet consists of multiple mixing servers arranged in sequence, each operated by an independent mixing authority, ensuring that as long as at least one server is honest, the confidentiality of the ballots is preserved. The mixnet operation on the input encrypted ballot set  $\{\enc(\mathbf{b}_{\mathbf{i}})\}_i$ is denoted by $\mix(\{\enc(\mathbf{b}_{\mathbf{i}})\}_i)$.

    We also use NIZK proofs to demonstrate  that the mixnet input $\mathcal{I}$ (a set of vote vectors) is correctly shuffled and re-encrypted to $\mathcal{O}$, without revealing any information about the permutation from $\mathcal{I}$ to $\mathcal{O}$. This proof is denoted by $\MixZKP{\mathcal{I}}{\mathcal{O}}$

\paragraph{Digital signature} The digital signature on message $m$ by voter $i$ is represented as $\sig_{i}(m)$. Most often we produce both the message and its signature, which we denote by $\auth_{i}(m)$, meaning $\sig_{i}(m)$ concatenated with $m$.

\paragraph{Hash function} $\hash(\cdot)$ denotes the hash of message $m$ using a collision-resistant cryptographic
hash function (e.g., SHA-256).

\paragraph{Serial Number} Each ballot includes a unique serial number $sn$ which is randomly generated. The role of the serial number is to allow one-to-one matching between the paper ballot and its electronic counterpart.

    \subsection{Risk Limiting Audit} \label{subsec:RLAsAndOneAudit}
    A \emph{Risk Limiting Audit} (RLA) tests whether a trustworthy set of paper ballots implies that the announced election winner(s) won, and corrects the result by a full hand count if not. It is parameterized by a \emph{risk limit} $\alpha$ and has the following property:

    \begin{quote}
       If the announced election result is wrong, the RLA will progress to a full hand count with probability at least $1-\alpha.$
        \end{quote}

    RLAs were developed by Stark and others \cite{stark2008conservative,stark2008sharper, stark2010super, lindeman2012bravo, lindeman2012gentle, ottoboni2018risk, stark2020sets} and apply to a wide variety of election types and audit situations, including single- and multi-winner plurality contests, supermajority elections, party-list proportional elections and Instant Runoff Voting.

    The \emph{apparent outcome} is (a set of) announced winner(s).
    The \emph{actual outcome} is the outcome that would be found if all the ballot papers were correctly counted.
    A \emph{ballot manifest} is a catalogue describing which paper ballots are present at which physical location.
    The apparent outcome is supported by a set of \emph{cast vote records} (CVRs), which may or may not accurately reflect the voter's intentions.
    In this work we mostly concentrate on \emph{ballot-level comparison audits} in which each individually sampled ballot paper is compared with its corresponding CVR. (Using \protname{} with other audit styles is discussed in \autoref{sec:otherRLAStyles}.)
   If the CVR overstates a winner's tally compared with the paper record, it is an \emph{overstatement} (for example, if the CVR is a vote for the winner and the paper ballot is a vote for the loser, it is a two-vote overstatement). If the CVR understates a winner's  tally compared with the paper ballot, it is an \emph{understatement}.

    An RLA technique that is particularly useful in our setting is the phantoms-to-zombies approach devised by Ba\~{n}uelos and Stark \cite{banuelos2012limiting}. This is a way of dealing with ballots that appear in the manifest but cannot be found on paper, a problem that may occur frequently when \protname{} ballots are delayed or dropped in the mail.  The idea is to ``Pretend that the audit actually ﬁnds a ballot, an evil zombie ballot that shows whatever would increase the risk value the most.'' Ba\~{n}uelos and Stark prove this to be conservative, in the sense that the RLA property still holds, assuming that it would have held had the correct ballot paper been located. The correct evil zombie ballot might be slightly different depending on the kind of RLA, but is generally a vote for the highest loser, or possibly an imaginary vote for all the losers, or (in the case of separate comparisons between the winner and each loser) a vote for whichever loser is being compared.

There are various RLA approaches, each with unique properties, defined by different mathematical functions for deciding when to accept the outcome. The following definition outlines their general functionalities and risk limit properties.

\begin{definition}

    \textbf{RL Functionality}

    \textbf{Setup parameters:} An apparent outcome, a ballot manifest, a risk limit $\alpha$, and optionally the cast vote records.

    \textbf{Input:} A sample of ballots (their manifest ID and vote contents), optionally indexed by their  cast vote record.

    \textbf{Procedure:} Perform one of the following actions:
        \begin{itemize}
            \item Output ``accept'' and halt, or
            \item Output ``reject'' and halt, or
            \item Take a new sample of ballots, combine it with the previous sample, update risk calculations and repeat.
        \end{itemize}

    \textbf{Risk Limit Property:} For any setup parameters with an incorrect apparent outcome, the probability that the above procedure outputs ``accept'' over uniformly generated\footnote{Some techniques exist for specific kinds of non-uniform sampling, such as stratified sampling~\cite{ottoboni2018risk, spertus2022sweeter}, and in some cases it matters whether the uniform sampling is conducted with or without replacement.} random samples is at most $\alpha$.

\end{definition}

    It is important to understand that applying an RL Functionality does not necessarily imply running a valid Risk Limiting Audit. For example, if the ballot papers are not an accurate representation of the voters' intent, if they have not been securely stored, if the CVRs were not properly committed before the random sample was taken, or the sampled ballots were not honestly randomly generated, the audit may not actually limit risk. The exact definition of a valid risk limiting audit is outside the scope of this paper. We can, however, prove that our RL functionality maintains the risk limit property, which implies that it can be fitted in to an existing Risk Limiting Audit process (if a valid one is already being run).

    \begin{remark}
    \label{RL remark}
    The risk limit property must hold regardless of how the wrong outcome is constructed: the adversary may choose any margin, or any way to distribute the wrong CVRs, but must commit to CVRs before choosing the samples randomly.
    \end{remark}

    \begin{definition}
    \label{def:monotonic}
    An RLA functionality is \emph{monotonic} if
        for any given setup parameters with an incorrect apparent outcome and  any input sample set,
    an increase in the discrepancy\footnote{The difference between the voter-verified paper record and the reported electronic results for a specific ballot} value of any sampled ballot, while other discrepancy values remain fixed,
    does not change the output from 'reject' to 'accept'.
    \end{definition}

    \begin{remark}\label{rem:phantom-remark}
        All RLA functions in general use are monotonic,
        and we will assume monotonicity in this paper.
    \end{remark}

\iffalse

    Proposition \ref{prop:one-audit-extension} and our use of ONEAudit in our protocol rely on the following proposition, which is proven in \cite{stark2023overstatement} (though not explicitly stated as a proposition).

    \begin{prop} \label{prop:starkOneAudit}
    If we have an RLA procedure to test whether the net overstatement is greater than the margin using the real CVRs produced by the voting system, the same procedure can test whether the outcome is correct---with the same risk limit---if it is applied to \emph{overstatement net equivalent} CVRs instead. (Audit sample sizes might be quite different.)
    \end{prop}

    The following proposition is also true, with essentially the same proof.
    \begin{prop} \label{prop:ballotONE}
        If we have an RLA procedure to test whether the net overstatement is greater than the margin using the real CVR-ballot correspondences produced by the voting system, the same procedure can test whether the outcome is correct---with the same risk limit---if it is applied to \emph{overstatement net equivalent} ballots instead. (Audit sample sizes might be quite different.)
    \end{prop}
    \begin{proof}
        The same one-line proof of \autoref{prop:starkOneAudit} given in \cite{stark2023overstatement}, except that we change the composition of $\Sigma_i b_i$ instead of $\Sigma_i c_i$, while maintaining a constant sum.
        \end{proof}

   \fi

    \section{Threat model and trust assumptions} \label{sec:threatModelAndTrustAssumptions}

    Here, we define our trust assumptions as follows:
    \begin{enumerate}
        \item The local electoral authorities maintain a list of CAC IDs assigned to registered eligible voters within their county.
        \begin{itemize}
            \item Each voter's CAC card securely stores the corresponding private key.
            \item The CAC Certificate Authority issues trustworthy certificates linking each CAC ID to its public key(s).
            \item The local electoral authorities validate the CAC certificate for each signed ballot posted to the BB.\footnote{In principle, the certificates could be posted on the BB with the signatures. However, in practice CAC certificates contain significant personal information that precludes public distribution. We therefore need to assume that they are verified by local authorities but not the public.}
        \end{itemize}
        \item Each voter verifies that the paper ballot accurately reflects their intention.
        \label{assumption:intention}
        \item Each printed signature on the sticker is verified before sending, either by the voter or by some other trustworthy assistant at the \remoteVotingCenter.
        \label{assumption:signver}
        %\item The list of CAC IDs for \protname{} ballots that have been validly signed on the BB aligns with the identifier list of voters who attempted to vote using the \protname{} system.\label{assumption:AccurateVoterCount}
        \item The mailing addresses on the envelopes are correct\footnote{This assumption can hold if: the sender knows the intended destinations, there is a supply of preprinted trustworthy envelopes, or an authority verifies the address printed on the envelope by the BMD.}.
        \item The voter instructions have to come from some trustworthy source, e.g. a poster on the wall, and not from the BMD itself. 
\end{enumerate}
    The below additional assumptions are relevant only for privacy.
    \begin{enumerate}[resume]
        \item The voter remains hidden from others within the confines of the voting booth. The privacy adversary's visibility is restricted to observing
        \begin{itemize}
            \item the BB,
            \item other pieces of evidence (remembered or captured by the voter), which might be susceptible to forgery.
        \end{itemize}
        \item Fewer than a pre-determined threshold of tallying authorities are dishonest.
        \item At least one of the mixing authorities is honest.

    \end{enumerate}

    Integrity also depends on an accurate count of the number of voters who participated---the procedures for achieving this are described in \autoref{sec:VoteCount}.

    \subsection{Attacker model}
    \label{sec:attacker model}
    We define two distinct adversarial models. In both models, the adversary may control all computers (except at least one chosen for verification) and do any polynomial-time computation. The variations between the models lie in the extent of control the adversary possesses over the paper channel.
    \begin{itemize}
        \item \emph{\nodelaynostuff} (or \Eonly): The adversary has no control over the paper channel. Every envelope arrives at the Local Counting Center unaltered within a specified timeframe and stuffing new envelopes is impossible.
        \item \emph{\nostuff} (or \EandD): Any envelope, mailed by a legitimate sender, may be selected by the adversary to experience extended delays or even go missing. However, if delivered, the envelope and its contents are unaltered.
    \end{itemize}
    The adversary in these models is unable to read the contents of a mailed envelope. Table \ref{Tab:adv} summarizes these models.

    If ballots are collected by a responsible authority (for example, in a ballot box) and either  transported  to the \localVotingCenter\ under supervision or audited on the spot with the \remoteVotingCenter\ acting as its own \localVotingCenter, this corresponds to the \nodelaynostuff\ model.

    The \nostuff\ model probably corresponds most closely to intuitive assumptions about postal mail voting (whether those assumptions are valid in practice or not). Jurisdictions vary greatly in how they prevent fraudulent ballots from being accepted. Our protocol's use of digital signature stickers makes it harder for an \emph{external} attacker to generate apparently-valid fake votes, but we assume that the attacker controls all the computers, including those to which the voter's CAC card is attached. Consequently, the attacker can manipulate a compromised voting machine with the voter's CAC card to sign any arbitrary message. We therefore do not make it harder for this \emph{internal} attacker to stuff a valid-looking paper ballot into the paper mail than it would have been with traditional postal mail. If the jurisdiction was diligently checking each voter's handwritten signature, or if some other form of registered or secure mail is being used, our protocol does not undermine it, but nor do we add to the defences. So we simply assume that \emph{something} prevents the internal attacker from adding fraudulent mail ballots.

    \begin{table}
        \begin{center}
            \caption{Capabilities of the adversary in  different  models.}
            \label{Tab:adv}

            \begin{tabular}{||l||c|c|c|c||}
                \hline
                \hline

                Adv.&Electronic&Env.&Env.&Content\\
                Model&Control&Stuffing&Drop&View\\

                \hline
                \Eonly\ &$\checkmark$&-&-&-\\
                \EandD&$\checkmark$&-&$\checkmark$&-\\
                \hline\hline
            \end{tabular}

        \end{center}
    \end{table}

    Note that \protname{} does not claim to defend against an attacker who can stuff paper ballots.

    \subsection{Implications}
    Our protocol strives to detect any errors that may occur during the voting process. However, complete recovery from these errors cannot be guaranteed. It is important to acknowledge that if any assumptions fail, the errors or attacks in the given adversarial model may become undetectable or it might be impossible to differentiate one type of failure from another.
    For instance, if the voters do not carefully check that their paper matches their intention, or do not verify that the digital signature they are about to place in the mail is valid (and theirs), then manipulation may be undetectable.

    \section{\protname{} protocol} \label{sec:cup-cake}
The \protname{} protocol is described below.
    The example given here uses a BMD to print a paper record, but the protocol would work just as well with a hand-marked paper ballot interpreted by a scanner. In that case, the ``voting computer'' would be a scanner with the capacity to produce an encrypted record and upload it to the BB.
     The serial numbers could either be generated by the scanner, or printed in advance.
    We will use the term ``voting computer'' or ``voting machine'' to mean either a BMD or a scanner with computation and communication capacity.

\subsection{Voting process for voter $i$}
\label{cup-voting-proc}
 The voting process from the voter's perspective is described below and shown in Fig.~\ref{fig:workflow}.
        
\begin{enumerate}
  \item The voting machine displays the options to voter $i$.

    \item The voter interacts with the machine to make their selection. Let their vote be $\mathbf{b}_i$.
    \item \label{step:make-vote} The voting machine assigns a unique serial number $sn_i$ to this vote. The voting machine produces:
	\begin{itemize}
		\item a signed, encrypted electronic record and serial number, with a proof of ballot validity, for the BB:
		$$\BB \auth_{i}(\enc(sn_i),\enc(\mathbf{b_i}), \ValidZKP{\mathbf{b_i}}),$$

		\item a plain paper record of the ballot $(sn_i,\mathbf{b_i})$, which is placed in an envelope for mailing, and 
		\item a sticker that includes an address, a traditional mail tracking number, a space for a pen-and-ink signature (if required by regulation---not relevant to our cryptographic protocol) and a digital signature
		$$\auth_{\textit{Voting Computer}} (Id_i,Id_e,H_i),$$
  where $Id_i$ is voter $i$'s unique identifier,  $Id_e$ is the election identifier, and $H_i$ is the hash value created from the digital record that is stored for voter $i$, i.e. $H_i=\hash(\auth_{i}(\enc(sn_i),\enc(\mathbf{b_i}), \ValidZKP{\mathbf{b_i}})).$
	\end{itemize}

The voter must verify that her plaintext printout matches her intention, then stick the sticker on her envelope and mail it. Someone also needs to use an electronic device, whether their own or a third party's, to  
\begin{itemize}
    \item verify the digital signature and the signed data on the sticker, and
    \item perform the ``BB Inclusion Check'' or ``Sticker BB Upload'' stated in \autoref{sec:VoteCount}.
\end{itemize}

\end{enumerate}

\begin{figure}
	\begin{center}

 \includegraphics[scale=0.6]{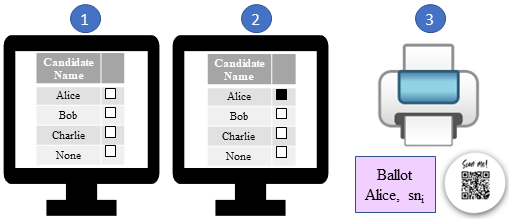}
  	\caption{The voting process.}

	\label{fig:workflow}
 \end{center}
\end{figure}

After the voting period ends, electronic records undergo a series of processing steps. Additionally, paper ballots are sent to the \voteCollectionCenter\ for auditing purposes. A diagram illustrating the whole process is in Figure~\ref{fig:scheme-sn}. Details of all processing steps will be provided in the following sections.

\begin{figure*}
	\begin{center}
 \includegraphics[scale=0.6]{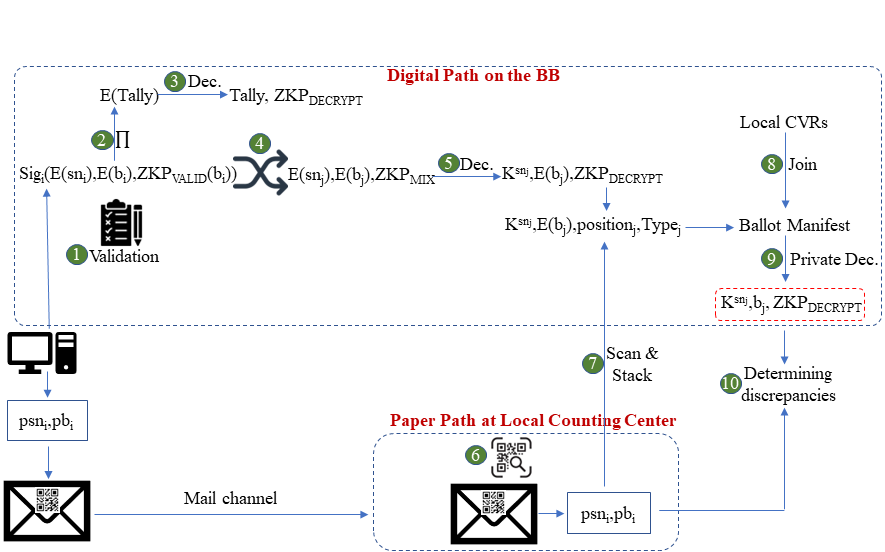}
	\caption{The whole process for \protname{}\ based on ballot comparison. Each \localVotingCenter's votes are dealt with separately---the diagram shows the process for only one \localVotingCenter.}
	\label{fig:scheme-sn}
 \end{center}
\end{figure*}

\subsection{Digital path}
\label{dig-path}
We explain digital processing of electronic ballots for one \localVotingCenter{}'s votes. Every other \localVotingCenter{}'s data is processed similarly---they do not interact because their results need to be delivered to the appropriate \localVotingCenter, even for statewide contests.

At the conclusion of the voting process, each jurisdiction processes the votes on the BB corresponding to CAC IDs registered in that jurisdiction.

The digital record undergoes the following processing steps.
\begin{enumerate}
       \item The validity of the signature and ZKP for each electronic ballot on the BB is verified.

       It is also verified that no two ballots on the bulletin board contain a common ciphertext. If this occurs, it
       indicates serious misbehaviour by a voting computer---the process stops.

       If any two valid ballots come from the same CAC ID, the first one is selected and the rest are withheld from any further processing.\footnote{It does not matter what definition of `first' is applied, as long as only one is chosen, and all $n$ selected for processing by a given jurisdiction are unambiguously marked on the BB.}
       \label{step1-dig}
    \item
       The digital record on the BB now contains the information of all $\numvoters$ voters, where the record with index $i$ denotes the ballot from voter $i$, where $i=1\ldots{}n$.
       $$\BB \mathcal{D}_1:=\auth_{i}(\enc(sn_i),\enc(\mathbf{b_i}), \ValidZKP{\mathbf{b_i}))}:i = 1\ldots{}n$$
    \item The encryption of the total tally is calculated and published on the BB

       $$\BB \enc(\mathbf{Tally}):=\prod_{i=1}^\numvoters{}\enc(\mathbf{b_i}).$$ \label{step:digitalEncryptedTally}
    \item $\enc(\mathbf{Tally})$ is decrypted and published on the BB  along with the proof of the correct decryption.

       $$\BB \enc(\mathbf{Tally}),  \mathbf{Tally}, \DecryptZKP{\mathbf{Tally}}{\enc(\mathbf{Tally})}$$ \label{step:digitalDecryptedTally}

       (In some jurisdictions, this value may be sensitive because of a small set of \protname{} voters.
       See \autoref{sec:otherRLAStyles} for variations that allow public verifiability without publication of the separate tally.)
    \item The votes and serial numbers are extracted from $\mathcal{D}_1$ and then mixed---denote this by $\mathcal{D}_2:=\enc(sn_i),\enc(\mathbf{b_i}): i=1\ldots{}n$. The mixed pairs along with the proof of the correct mix operation are published on the BB. So, at the end of this step, we have the following data on the BB:
\begin{align*}
&\BB \mathcal{D}_3:=(\enc(sn_j),\enc(\mathbf{b_j})):j=\pi(i),j=1\ldots{}n.\\ \nonumber
&\BB \MixZKP{\mathcal{D}_2}{\mathcal{D}_3} \nonumber
\end{align*}

    where $\pi$ is the (secret) permutation applied by the Mixnet. \label{step:mix}
    \item Serial numbers are decrypted and then the proof of the correct decryption is published. Therefore, at the end of this step, we have the following data on the BB.
\begin{align*}
    &\BB \\ \nonumber
    &\mathcal{D}_4:=(K^{sn_j},\enc(\mathbf{b_j}),\DecryptZKP{K^{sn_j}}{\enc(sn_j)}):j=1\ldots{}n.
\end{align*}

    It is also checked that there are no duplicate serial numbers. Any occurrence of duplicate serial numbers indicates significant problems with a voting computer---if there are any, the process should stop.
	\label{step:decryptDigitalSerialNumbers}

 \label{step:lastStepOfElectronicProcess}

\end{enumerate}

\subsection{Paper path} \label{subsec:paperPath}

At the end of the voting process, the ballot papers are sent to the \voteCollectionCenter\ individually in separate  envelopes. Each paper record undergoes the following processing steps.

\begin{enumerate}
	\item[6.]  If the jurisdiction has an existing process of scanning traditional handwritten signatures, this would  apply at this point. Envelopes with invalid signatures must be set aside and handled properly.
 
Each incoming envelope's sticker is scanned and its corresponding digital signature is verified. It is necessary to check that the data being signed matches the corresponding data on the BB for this voter, that the voting computer's signature is valid, and that the CAC ID is registered to vote in this jurisdiction. Then, the envelope is accepted and its corresponding record on the BB is marked as received. If any of these verifications fails, the process continues but the envelope is set aside unopened in a stack called ``Rejected Envelopes'' for further investigations. If the signature is valid but another validly signed envelope has already been received from the same voter, it is set aside in a stack called the ``Eligibility Problem'' stack.

   For all validly signed envelopes (excluding the ones inside the ``Eligibility Problem'' stack), the envelope contents are removed and physically shuffled, similar to standard postal voting procedures.
 \label{s:scan-papr}

    \item [7.] \label{step:check-serial-num} Each incoming ballot's serial number $psn$ is scanned. For each serial number, $K^{psn}$ is computed and it is  checked if there is a  digital ballot with a matching  $K^{sn}$ on the BB.\footnote{Because of the way ElectionGuard encryption works, the number that naturally drops out after decryption is $K^{sn}$, not $sn$. Because we care only about exact matches, it does not matter whether we work in the exponential or plain form. We do not assume that $K^{sn}$ hides $sn$, and do not need to hide the values of $sn$, because these are not publicly associated with the voter.} Then, the location of the paper ballot is appended to the corresponding digital ballot on the BB.\footnote{This is its physical location in paper ballot storage, as in a ballot manifest. For example, ``23rd ballot in batch 7, cabinet 12.''} Each digital ballot is accompanied by a parameter called $\Type$ where $\Type=\text{``not matching''}$, $\Type=\text{``one-to-one''}$ and $\Type=\text{``duplicated''}$ respectively indicate if there exist zero, one or multiple paper ballots with the matching serial number. Therefore, we have the following data on the BB:
    
    \begin{align*}
        \BB& \mathcal{D}_5:=\{(K^{sn_j},\enc(\mathbf{b_j}),\{\pos{}_{j'}\}_{j'},\Type_j):\\ \nonumber
        &j=\pi(i),j'=\pi'(i),j=1\ldots{}n.\}\nonumber
    \end{align*}

The permutation  $\pi'$ represents a physical shuffle of paper ballots, and $\pos{}_{j'}$ denotes the location of the associated paper ballots in storage. For digital ballots with $\Type=\text{``duplicated''}$ the positions of all matching ballots are recorded on the BB. For digital ballots with $\Type=\text{``not matching''}$, the parameter \pos{} is null.

Any paper ballot with no matching $K^{sn}$ on the BB is set aside. Call this the ``paper-only'' stack.  These need to be dealt with outside the cryptographic protocol.

\label{s:scan-papr2}

\end{enumerate}

\subsection{RLA}
\label{subsec:RLA}

\begin{enumerate}
    
\item [8.] At the \localVotingCenter\ the usual local Cast Vote Records (CVRs) are joined with the \protname{} CVRs for auditing purposes, and the ballot manifest is updated to include both kinds of records. When a local ballot is sampled,  it is retrieved and dealt with as usual. When a \protname{} ballot is sampled, its discrepancy is determined according to the guidelines outlined below. Then RLA statistics are updated accordingly, in exactly the same way for remote ballots as they would be for local ones.

\end{enumerate}

    Local officials make local decisions about whether to accept the result or escalate to a larger sample, according to whatever RLA calculations they usually conduct.

\paragraph{Determining discrepancies}
\label{discrepancy-det}
 Assigning  discrepancies to digital ballots is  primarily determined by the stack to which the ballot belongs.   

\begin{enumerate}
    \item [9.1.] If the type for the selected ballot is either $\text{``one-to-one''}$ or $\text{``duplicated''}$, its encrypted vote $\enc(\mathbf{b})$ is privately decrypted to $\mathbf{b}$ and, alongside the proof of correct decryption, is supplied to all auditors and observers in the \localVotingCenter, who verify the NIZKPs. They have access to the following data for the selected ballot:

$$\textbf{Observers: } K^{sn_j},\enc(\mathbf{b_j}),\mathbf{b_j},\DecryptZKP{\mathbf{b_j}}{\enc(\mathbf{b_j})},\{\pos{}_{j'}\}_{j'}$$

Then, its discrepancy is determined as follows.
\begin{itemize}
    \item Identify the paper ballot situated at position $\pos{}_{j'}$ and read its serial number $psn$. If $K^{psn}$ does not match $K^{sn_{j}}$, output ``serial number error'' and terminate.
    \item Compare the vote on the paper $\mathbf{pb}$ to $\mathbf{b_{j}}$ and record the discrepancy, exactly like any other RLA.
\end{itemize}
For ballots of type $\text{``duplicated''}$, the above 2-step procedure is performed for all $j'$ values and the final discrepancy is set to the \emph{minimum} discrepancy among different $j'$ values for the digital ballot. In other words, in case
of several paper ballots, one of the ballots with lowest discrepancy value is selected.

Any occurrences of the ``serial number error'' indicates a significant problem with the scanning device and hence necessitates restarting the entire process from step 7 onward using another scanning device.

\item [9.2.] For ballots with $\Type=\text{``not matching''}$, the encrypted ballot is privately decrypted and, alongside the proof of correct decryption, is supplied to all auditors and observers. Then, the discrepancy is determined by setting it to the maximum possible value for such a ballot, or in other words, employing a worst-case paper assumption for the given ballot.

\begin{remark}
\label{remark:oneaudit}
    An alternative  to handle ballots with $\Type=\text{``not matching''}$ is based on ONEAudit as follows. All encrypted ballots with $\Type=\text{``not matching''}$ are homomorphically added and its tally is privately decrypted for the auditors and observers in the \localVotingCenter.  Each individual CVR is then taken to be the average \emph{overstatement net equivalent} CVR\footnote{For example, given that 4 out of 10 'not matching' votes were cast for the first candidate and 6 out of 10 for the second, the resulting vote shares are 0.4 and 0.6, respectively.}. The (absent) paper ballot is taken to be the worst-case ballot.
\end{remark}
\iffalse
\begin{itemize}
    \item \nodelaynostuff\ model: Any envelope in the ``Eligibility Problem'' stack demonstrates a major problem that halts the process.
    \item \nodelay\ and \nostuff\ models: For any validly signed envelope in the ``Eligibility Problem'' with a distinct voter's identifier, a maximum discrepancy (typically +2) is added to their county votes.

\end{itemize}
\fi

\end{enumerate}

    Note that this system does \emph{not} allow for individual recorded-as-cast verification: individual
    voters do not get evidence that their electronic record matches their intention. This evidence---or rather, evidence
    about whether the discrepancy is large enough to alter the result---is provided collectively
    through the matching and subsequent audit processes.

    A summary of all verification steps is provided in Appendix~\ref{subsec:verification}, with details of error handling in Appendix~\ref{subsec:failure-handling}.

\subsection{A simpler design for the \nodelaynostuff\ model}
If, instead of mailing envelopes individually, ballots are collected in a ballot box and conveyed in an organized manner, the voting process would differ, as no sticker is provided to the voter and  ballots are not placed in envelopes. Consequently, voters solely verify their ballots and place them in a ballot box. Furthermore, in the ballot box scenario, step 6, which is primarily relevant to sticker and envelope processing, will be omitted. However, it is crucial to verify that the number of ballots in the box for each ballot style does not exceed the corresponding quantities on the BB.  For \nostuff\ model this condition is automatically met through step 6. The whole digital path and Step 7 in the paper path remain consistent with the details outlined in Sections~\ref{dig-path} and~\ref{subsec:paperPath}, respectively. The RLA process is also the same as that outlined in Section~\ref{subsec:RLA}.

    \subsection{Counting the number of votes and creating a trustworthy ballot manifest}
    \label{sec:VoteCount}

    As voters cast their ballots using their CAC cards, the digital signature associated with each vote on the BB offers a dependable method to authenticate the voter.
    The purpose of the digital signatures is to complicate ballot stuffing---only ballots (either electronic or paper) with a valid accompanying digital signature will be accepted.  The voting computer's digital signature needs to be explicit on the mail sticker, but is implicit on the BB because that channel is already authenticated.

    However, there is still the possibility of an attack in which a malicious voting computer colludes with some other machine that (at some point) has CAC access, in order to fabricate an apparently-valid digital vote from a voter
    who did not intend to vote (and did not appear at the polling place). Since this is electronically indistinguishable from valid voting, it must be defended by human procedures.
    \begin{itemize}
        \item Officials at the \remoteVotingCenter\ must keep a count of the
        total number of properly completed votes.
    \end{itemize}
    It is a requirement of any RLA that there must be a record of how many ballot papers there are, derived independently of the scanners that are being audited. Otherwise, if the paper records are also delayed or dropped, there is no way to know how many ballots have been dropped, and consequently what the implications for the accuracy of the election result might be.

    There also needs to be a way to verify that ballots have not been dropped  from the electronic record.
    This needs to be supported by one of the following practical procedures.
    \begin{itemize}
        \item Officials at the \remoteVotingCenter\ could keep a record of who voted there and verify that it (eventually) matches the set of valid signatures on votes on the BB.
        \item ``\BBInclusionCheck{}'': Every voter (or their delegate) would need to verify the inclusion of their vote on the BB. The envelope sticker can be used as participation evidence.
        \item ``\StickerBBUpload{}'': Every voter (or their delegate) would upload the sticker's data to the server and receive a signed receipt. When possible, the server would send it to the BB. Then, everyone can check the presence of a corresponding vote on the BB for every uploaded sticker.

    \end{itemize}
    Options~2 and~3 are variations of the same concept: the actual BB upload might take a long time, so we ensure that the voter gets an ``inclusion promise'', which can be used as evidence that she voted. In the case of ``\BBInclusionCheck'', the envelope sticker itself is the ``inclusion promise''; in the case of ``\StickerBBUpload'', the server's signature on that value is an inclusion promise from the server. This allows for evidence of different kinds of malfeasance from different parties, and requires different forms of verification.
    Using either of the stated methods to process the sticker's data does not compromise privacy  and may be done by voters, officials or bystanders.

In the \nodelaynostuff\ model, assuming each voter verifies that their ballot aligns with their intentions, the paper manifest accurately reflects the legitimate voters' intentions. Since paper ballot drop is impossible in this model, digital ballot drop is effectively prevented if an authority ensures that the number of paper ballots in the ballot box does not exceed the number of digital ballots recorded on the BB.

    Table~\ref{Tab:verification} summarizes the verification steps required by the voter under various adversarial models.

    \begin{table}
        \begin{center}
            \caption{Verification steps required by the voter under various adversarial models.}
            \label{Tab:verification}

            \begin{tabular}{||l||c|c|c||}
                \hline
                \hline

                Adv.&Voter Intent&Sticker Sig.&BB Incl. Check or\\
                Model&Verification&Validation$^{\dagger}$ &Sticker BB Upl.$^{\dagger}$\\

                \hline
                \Eonly\ &$\checkmark$&-&-\\
                \EandD&$\checkmark$&$\checkmark$&$\checkmark$\\
                \hline\hline
            \end{tabular}
            {\small \item $^\dagger$: Also feasible for an official at the \remoteVotingCenter}
        \end{center}
    \end{table}

    \section{Security analysis}
\subsection{Verifiability}\label{subsec:security-verifiability}

We begin by outlining the audit process in an ideal scenario featuring trustworthy paper ballots and a one-to-one correspondence between each paper ballot and its electronic record. A specific subset of voters in the ideal world is controlled by the adversary, but other voters follow the voting instructions.

Next, we progressively transform this ideal world into our protocol (i.e., the Real World) through a series of intermediate worlds. At each stage, we establish the validity of the following  statement: \emph{any combination of digital and paper ballots in the newly defined world corresponds to a set of digital and paper ballots in the present world. This correspondence ensures identical (1) reported totals, (2) number of voters, and (3) honest voters. Furthermore,}

\begin{itemize}
      \item   \emph{The digital ballots in the newly defined world are a permutation of those in the present world, or}
      \item \emph{For any ballot selected for the RLA in the present world, the discrepancy value is equal to or smaller than the discrepancy value of its corresponding ballot in the newly defined world.}
    \end{itemize}

    These statements, based on Remarks~\ref{RL remark} and~\ref{rem:phantom-remark} respectively, imply that for any given risk limit, when the announced outcome is wrong, if the RLA functionality with any given sample set accepts the announced outcome in a newly defined world, it will accept the corresponding outcome in the present world. The intuition is that inconsistencies between CVRs and paper records are handled in such a way that the overall risk value does not decrease from one scenario to the next.

    Since the last defined world is the same as the real \protname{} protocol, the proof shows that it has no greater likelihood of accepting a wrong result than an `ordinary' RLA conducted in the ideal world on the trustworthy paper ballots with a one-to-one correspondence between each paper ballot and its CVR.

\paragraph{Ideal World}
Consider an Ideal World where every voter  visits a local polling place to verify their paper ballot and puts it in a ballot box. A subset of voters is under the control of the adversary. Consequently, the adversary determines the participation status of each in the voting process. The adversary can also influence these voters, leading them to deviate from the prescribed voting process instructions.

The voting machine records a digital ballot for each voter, where the validity of the digital ballot content is publicly verifiable. The resulting CVRs are accurately incorporated into the local CVRs where authorities at the \localVotingCenter{} verify the correctness of this operation (i.e., Section \ref{sec:Ver-at-voteCollectionCenter}, Verification at the \voteCollectionCenter, item \ref{item:tally-added-to-CVRs}).

In this ideal world, the voter's identifier is recorded on both the paper and its digital record. A precise one-to-one correspondence exists between each paper ballot and its electronic record. This correspondence is achieved relying on  a trusted authority who verifies whether the digital ballot for the voter appears on the BB before allowing the voter to put their ballot in the box. In other words, the trusted authority ensures the fulfillment of two critical requirements: (1) the identifiers printed on the paper ballots form a subset of the identifiers recorded on the BB (ensuring inclusion), and (2) the quantity of votes in the ballot box is equal to the number of ballots recorded on the BB (ensuring exclusion). We assume the voting authority also has a way to verify each voter's identity.

The adversary has control over all electronic devices, including the voting computer. Therefore, the voter's selection on each paper ballot  might be different from its corresponding digital record. After tallying the digital ballots and reporting the voting outcome, we need to verify if discrepancies between digital and paper ballots are not enough to change the voting outcome. So, we execute an RLA process  by unsealing the ballot box and randomly sampling from the paper records and using the voter's identifier printed on the chosen paper to locate its corresponding digital record.  Then, discrepancies are properly determined and then correctly incorporated into the RLA statistics where authorities at the \localVotingCenter, verify the correctness of the operation (i.e., Section \ref{sec:Ver-at-voteCollectionCenter}, Verification at the \voteCollectionCenter, item \ref{item:disc-RLA-correctness}).  To guarantee that an incorrect voting outcome will lead to a manual tally with a predefined risk limit, each voter must check if the selections, printed on their paper ballot, is correct (i.e., Section \ref{sec:Cast-as-intended verification}, Cast-as-intended verification by the voter, item 1).

  \paragraph{Ideal World Game}
    The adversary $\mathcal{A}$ wins the ideal-world game if in the Ideal World,
    \begin{itemize}
        \item the apparent election result, according to the (\protname{} and regular) CVRs, is different from the actual result, according to the (\protname{} and regular) paper ballots, and
        \item the RLA procedure with a configured risk limit $\alpha$ (including the \protname{} cryptographic verification) confirms the result with a probability non-negligibly greater than $\alpha$, where the probability is taken over the random ballot selections of the RLA, randomization in cryptographic algorithms, and the adversary's random choices.
        \end{itemize}

\paragraph{Real World}
This world resembles our protocol in the \nostuff\ or \nodelaynostuff\ model, with an adversary with the following capabilities:
\begin{itemize}
    \item Controlling a subset of voters, causing them either not to vote or to deviate from established procedures.
    \item Printing arbitrary serial numbers on paper ballots.
    \item Controlling the voting computer, thereby recording digital records that are inconsistent with the voter's selections. 
    \item Dropping some envelopes in the \nostuff\ model.
    \item Controlling the scanning device (used in step \ref{step:check-serial-num} of the protocol to scan paper ballots' serial numbers), causing it to record  arbitrary serial numbers.
\end{itemize}

The Real World Game is defined similarly to the Ideal World Game, with the distinction that it is set in the Real World. 

    Our main result is that \protname{} does not give the adversary a non-negligibly higher chance of causing the RLA to accept a wrong result, than running the normal RLA without \protname.
\begin{coro}\label{coro:main-result}
    Assuming all verification steps (outlined in Appendix~\ref{subsec:verification}) are successfully completed, for any adversary $\mathcal{A}$ who wins the Real World Game, there exists an adversary $\mathcal{A}'$ who wins the Ideal World Game.
\end{coro}

\begin{proof}
    The proof for the \nostuff\ model follows directly from Propositions \ref{prop:BMD sig} to \ref{prop:random-dig}. The proof for the \nodelaynostuff\ model follows directly from Propositions \ref{prop:exclusion} to \ref{prop:mal-scan} and Propositions \ref{prop:paper-path-4-to-3} to \ref{prop:random-dig}. These propositions are proven in Appendix~\ref{subsec:verifiability-proof}. They gradually transform a real-world attack into an equally successful ideal-world attack.
\end{proof}

\subsection{Privacy} \label{subsec:privacy}

    \protname{} is intended to run as part of a larger protocol in which most ballots in the election are input by other means.
    However, privacy is difficult to model in that setting, so here we treat \protname{} as a standalone protocol.

    \protname{} in the \nostuff\ model does not provide ballot privacy, because the adversary in this model might drop envelopes and thus use differencing attacks to infer individual votes. For example, if the adversary drops one ballot, labelled on the sticker as Alice's, then the adversary knows whose vote is the unique one with no match. Since there is a non-zero probability that this is audited and opened, Alice's privacy cannot be guaranteed. More generally, since the RLA process may lead to a full manual recount of the arrived paper ballots, differencing attacks are always possible.

    Our privacy proof therefore assumes the \nodelaynostuff{} model---all paper ballots are guaranteed to arrive. It is the RLA step that allows for differencing attacks. Since the digital paths in both models are the same,  \protname{} in the \nostuff\ model excluding the RLA components provides privacy also.

    Opening ballots for the RLA inevitably has some privacy implications, regardless of whether the ballots are mixed cryptographically or physically. \protname{} neither introduces nor solves these problems.
    For example, if range voting is used with only two voters, then the attacker can distinguish a situation in which Bob votes 2 and Alice 0, from one in which Bob votes 1 and Alice 1, if either ballot is selected for audit, though the tallies are the same.

    Even when no ballots are audited, small anonymity sets may reveal individuals' preferences, for example if there are only a very small number of remote votes sent to one \localVotingCenter, and they all make the same choices.
    Group privacy of the whole set of remote voters may also cause concern, for example if the collective choices are strongly skewed relative to the rest of the population.

    Table \ref{Tab:privacy-voting} summarizes what is visible to various participants.
    Access to both the voter's identifier and their vote is limited to the voting computer.
    The main privacy result is that the protocol does not reveal more about an individual voter's intent than the
    published mixed votes do.

    \begin{table}
        \begin{center}
            \caption{Data accessible to various entities}
            \label{Tab:privacy-voting}

            \begin{tabular}{||l||c|c|c|c|c||}
                \hline
                \hline

                Entity&ID&SN&ID\&SN&Vote&Tally\\
                \hline
                General public&$\checkmark$&$\checkmark$&$\times$&$\times$&$\checkmark^{\dagger}$\\

                Voting computer&$\checkmark$&$\checkmark$&$\checkmark$&$\checkmark$&$\checkmark$\\
                RLA observers&$\times$&$\checkmark$&$\times$&$\checkmark$&$\checkmark^{\dagger}$\\
                Postal worker&$\checkmark$&$\times$&$\times$&$\times$&$\checkmark^{\dagger}$\\

                \hline\hline
            \end{tabular}
            {\small \item $^\dagger$: visible on the BB}
        \end{center}
    \end{table}

    \begin{definition}
        The privacy attacker may view the following:
        \begin{enumerate}
            \item the complete contents of the BB, including vote tallies,
            \item envelope stickers arriving at the \localVotingCenter,
            \item paper ballots that arrive at the \localVotingCenter, after they have been physically disassociated from their identifying envelope,
            \item decryptions (and proofs of proper decryption) of mixed votes that are sampled for audit,
        \end{enumerate}
        and may control:
        \begin{enumerate}
            \item all voting computers for corrupted voters,
            \item all but one of the mixing servers,
            \item a sub-threshold number of decryption authorities.
        \end{enumerate}
        Subverting the privacy-preserving process of paper ballot handling (e.g. by opening the named envelope and reading its contents), is outside the attacker model.
        The privacy attacker may not open vote envelopes, and does not see honest voters' serial numbers or votes during the voting phase.
    \end{definition}

    We use an instance of the privacy definition BPRIV from~\cite{bernhard2015sok}, in which the algorithms for
    Setup, Vote, Valid, Publish, Tally and Verify are derived from the \protname{} protocol. Importantly, the \emph{Tally} protocol is defined to publish both the decrypted vote tally from Step~\ref{step:digitalDecryptedTally} and also the random permutation of individual votes from Step~\ref{step:lastStepOfElectronicProcess}.

    The attacker $\adv$ is allowed to corrupt some voters, and is trying to guess how the honest ones voted. $\adv$ is allowed to place different (encrypted) votes for each honest voter on $BB_0$ and $BB_1$. The challenger will choose $\beta \in \{0,1\}$ at random. If $\beta=0$, the adversary is shown the contents of $BB_0$, the tally from $BB_0$ and a validly-generated proof of mixing and decryption to produce that tally; if $\beta=1$, the adversary is shown the contents of $BB_1$, \emph{the tally from $BB_0$} and a simulated proof that $BB_1$ produces that tally (which it may not do in reality). If the adversary cannot distinguish these two transcripts, it cannot guess how individuals voted.

    Our proof relies upon privacy properties of both the underlying ElectionGuard cryptographic library and the mixing ZKP.
    In both cases, these assumptions seem reasonable though we are not aware of complete proofs in the literature.

    \begin{prop}\label{prop:main-privacy}
    \protname{} in the \nodelaynostuff\ model provides ballot privacy according to BPRIV~\cite{bernhard2015sok}, assuming that the underlying ElectionGuard ciphertexts satisfy NM-CPA and the mixing proof is zero knowledge.
    \end{prop}
            The (adapted) definition of BPRIV, definitions of the relevant functions, and the proof, are in Appendix~\ref{subsec:privacy-proof}.

\subsection{Receipt freeness} \label{subsec:security-receipt-freeness}

    We assume the presence of a coercer outside the voting booth who is already aware of the target voter's identity and requests their serial number. Given that (1) serial numbers are ultimately disclosed on the BB and (2) serial numbers are sufficiently long, the voter is unable to provide false information about their serial number. Nevertheless, the coercer cannot rely on the voter's honesty regarding their cast vote.
In other words, the voter cannot convince the coercer that the encrypted value $\enc(\mathbf{b_i})$ printed on the the sticker or published on the BB before the mixnet or the $\enc(\mathbf{b_i})$  in the \remoteCommittedCVR\ (i.e. $(sn_i,\enc(\mathbf{b_i}))$ after the mixnet) is decrypted to the $\mathbf{b_i}$ that is claimed by the voter as her vote. However, as Table \ref{Tab:privacy-coercer} shows,  collusion between a coercer and auditors or observers at the \voteCollectionCenter\ could enable the use of serial number data to connect voters' identities with their votes.

    Appendix~\ref{subsec:batch} elaborates on the removal of serial numbers from our protocol. This ensures the preservation of voter privacy, even in the event of collusion between a coercer and auditors or observers at the \voteCollectionCenter.

\begin{table}
\begin{center}
  \caption{Data accessible to the coercer and the RLA observers.}
  \label{Tab:privacy-coercer}

\begin{tabular}{||l||c|c|c|c|c||}
 \hline
 \hline

  	Entity&ID&SN&ID\&SN&Vote&Tally\\
 \hline

Coercer&$\checkmark$&$\checkmark$&$\checkmark$&$\times$&$\checkmark^{\dagger}$\\
RLA observers&$\times$&$\checkmark$&$\times$&$\checkmark$&$\checkmark^{\dagger}$\\

\hline\hline
  \end{tabular}
{\small \item $^\dagger$: visible on the BB}
  \end{center}
 \end{table}

\section{Using \protname{} with other RLA styles, avoiding the publication of \protname{} sub-tallies} \label{sec:otherRLAStyles}
    If a \localVotingCenter{}\ has only a very small number of \protname{}\ ballots, it may not be acceptable to publish their tally in \autoref{step:digitalDecryptedTally}.
    When an attacker has direct access to the arriving mail ballot, we do not have a solution to this issue.
    However, if we consider a remote attacker who is looking at the sub-tallies on the BB, we can address the problem by including the \protname{}\ tally into a larger tally of ballots. In more details, the ordinary ballots from the \localVotingCenter{} can simply be appended to the \protname{}\ ballots and then the tallies are  computed over the whole set of votes. Appendix~\ref{App:otherRLAStyles} describes multiple variants that
     obviate the need to publish the sub-tally of \protname{}\ ballots in \autoref{step:digitalDecryptedTally} of the \protname{}\ digital path. Some also avoid printing individual serial numbers on ballots, relying on batch numbers that identify a group of ballots instead.  Appendix~\ref{subsec:MergeWithVault} describes how to incorporate \protname{} into audits using VAULT~\cite{benaloh2019vault}, retaining serial numbers, mixing, and individual ballot-level comparisons. Appendix~\ref{subsec:SubMerge} describes how to incorporate \protname{} ballots into larger batches and audit them as a batch---this requires neither mixing nor serial numbers, but may result in a less efficient audit.

   These variations may also be useful for localities that use either VAULT or batch-level comparison audits anyway.

\section{Implementation results}

In this section we discuss our prototype library for the \protname{} protocol.
The prototype library\footnote{Available at
    \anon{\url{https://github.com/JohnLCaron/egk-ec-mixnet}}}
    has an ElectionGuard 2.0 implementation,
    as well as Elliptic Curve (NIST P-256) and mixnet implementations derived from the Verificatum library\footnote{See https://www.verificatum.org/},
    along with an option to use GMP\footnote{See https://gmplib.org/} for the low-level computations, also derived from Verificatum.
    Note that this is a prototype library for proof-of-concept purposes, not a production-quality cryptographic library.
    For example, we have not (yet) implemented hash canonicalization as per the ElectionGuard specification.

    For the measurements in this section we used a fast CPU (Intel Xeon CPU E5-2680 v3 @ 2.50GHz), for ballots with
   7 contests, each having 4 selections.
    Each ElectionGuard contest has one more ciphertext than the number of real selections, so there are $\text{numberOfContests} + \text{totalNumberOfSelections} = 7 + (7 \times 4) = 35$ encryptions per ballot.
    Using elliptic curves, a single ballot like this can be encrypted (and the ZKPs generated) in about 20 ms.
    Encryption is done at the \remoteVotingCenter{}, which we expect to have modest hardware and be significantly slower.
    Ballots can be encrypted in parallel, which scales almost linearly with the number of available cores.

    The rest of the cryptographic processing is done on a server, typically at each Local Counting Center.
    The cost of homomorphic tallying  (.22 ms per encryption) and decryption of the tally (2.2 ms per encryption for 3 trustees) is probably negligible.
    The cost of ElectionGuard verification is about 15 ms per ballot before parallel speedup.

The largest time used in our workflow is for the mixnet and its verification.
    The election admin can run as many rounds of the mixnet as desired; each round consists of a shuffle/reencryption of the ballots and a generation of a proof of shuffle.
    Each round's proof must be verified, but only the last round's shuffle is used in the RLA.
    Using algorithms developed by Verificatum, one round of the mixnet takes about 60 msecs per ballot to run the mixnet and generate the proof.
    Verifying one round of the mixnet proof takes about 80 msecs per ballot.
    Both timings are single threaded on the fast CPU, and we get good speedup when running on multicore machines.

When doing ballot level RLA, the encrypted ballots selected for the audit must be decrypted using standard ElectionGuard threshold decryption.
    This costs about 2.3 msecs per encryption (including generating the decryption proof) for 3 trustees on the fast CPU, single threaded.
    Verifying the decryption proof costs about 15 ms per ballot.

  \section{Conclusion and future work}
    Usable voter-verification---without undermining coercion resistance---remains the major open problem in this field.
    Using paper requires transmitting it back to the counting location; using cryptography demands that the voter successfully performs one of the human-computer-interaction protocols which, despite ingenious recent advances, remain very hard for voters to run successfully.

    Everlasting privacy is an obvious important addition to \protname{}, which could probably be achieved in a standard way (by posting perfectly hiding commitments to the vote, rather than ciphertexts).

    Another interesting direction is to examine how to use cryptography to prevent modification of the paper ballot in transit, without impacting privacy. It is obvious that if the envelope sticker contained a signed hash of the \emph{plaintext} vote, then this could be verified upon arrival---but this breaks privacy. It is interesting to ask whether there is a privacy-preserving way of protecting the integrity of paper ballots in transit.

    %\newpage

    \section{Acknowledgements} We would like to thank Dan Wallach for valuable ideas in the early phases of this project. Thanks also to Josh Benaloh, Michael Naehrig, Olivier Pereira,
    John Sebes and Andrew Appel for helpful comments on earlier drafts of this paper, and to Thomas Haines for review of the mixnet code. The UI team at Rice and the development team at VotingWorks greatly improved the practicality and usability of the system. Finally, careful reviews by Trail of Bits and various anonymous reviewers significantly influenced this protocol for the better.

    \section{Open science}
    The only significant artefact of this work, other than the paper, is the implementation. All that code is openly available
    under open licenses, either GPL3.0 or the MIT license. Everything is available for allowing others to clone, run, edit and reuse our code.

    \section{Ethics considerations}
    Work on voting always requires significant ethics considerations. This protocol is designed for practical use, so we have to consider carefully how it might be used in practice. The main ethical consideration---which has been raised by informal reviews we sought before formal submission---is that we cannot guarantee the protocol is run in the conservative way we describe here. People do not always follow instructions, and might be tempted to count electronic votes without running the RLA in the rigorous way we have specified. We acknowledge this risk, but feel that this applies to any security protocol. We have been careful to specify as clearly as possible the procedural supports necessary to run this protocol securely.

    %\bibliographystyle{plain}
    %\bibliography{bibliography}

\begin{thebibliography}{10}

\bibitem{adida2008helios}
Ben Adida.
\newblock Helios: Web-based open-audit voting.
\newblock In {\em USENIX security symposium}, volume~17, pages 335--348, 2008.

\bibitem{appel2020ballot}
Andrew~W Appel, Richard~A DeMillo, and Philip~B Stark.
\newblock Ballot-marking devices cannot ensure the will of the voters.
\newblock {\em Election Law Journal: Rules, Politics, and Policy},
  19(3):432--450, 2020.

\bibitem{banuelos2012limiting}
Jorge~H Banuelos and Philip~B Stark.
\newblock Limiting risk by turning manifest phantoms into evil zombies.
\newblock {\em arXiv preprint arXiv:1207.3413}, 2012.
\newblock \url{https://arxiv.org/abs/1207.3413}.

\bibitem{bell2013star}
Susan Bell, Josh Benaloh, Michael~D Byrne, Dana DeBeauvoir, Bryce Eakin, Philip
  Kortum, Neal McBurnett, Olivier Pereira, Philip~B Stark, Dan~S Wallach,
  et~al.
\newblock $\{$STAR-Vote$\}$: A secure, transparent, auditable, and reliable
  voting system.
\newblock In {\em 2013 Electronic Voting Technology Workshop/Workshop on
  Trustworthy Elections (EVT/WOTE 13)}, 2013.

\bibitem{electionguard}
Josh Benaloh and Michael Naehrig.
\newblock Electionguard specification v2.0.
\newblock 2023.

\bibitem{benaloh-electionguard}
Josh Benaloh, Michael Naehrig, Olivier Pereira, and Dan~S. Wallach.
\newblock {ElectionGuard}: a cryptographic toolkit to enable verifiable
  elections, 2024.
\newblock Cryptology ePrint Archive, Paper 2024/955.

\bibitem{benaloh2019vault}
Josh Benaloh, Philip~B Stark, and Vanessa Teague.
\newblock Vault: Verifiable audits using limited transparency.
\newblock {\em Proceedings of E-Vote ID}, 2019.
\newblock \url{https://www.stat.berkeley.edu/~stark/Preprints/vault19.pdf}.

\bibitem{bernhard2015sok}
David Bernhard, V{\'e}ronique Cortier, David Galindo, Olivier Pereira, and
  Bogdan Warinschi.
\newblock Sok: A comprehensive analysis of game-based ballot privacy
  definitions.
\newblock In {\em 2015 IEEE Symposium on Security and Privacy}, pages 499--516.
  IEEE, 2015.

\bibitem{bernhard2012not}
David Bernhard, Olivier Pereira, and Bogdan Warinschi.
\newblock How not to prove yourself: Pitfalls of the fiat-shamir heuristic and
  applications to helios.
\newblock In {\em Advances in Cryptology--ASIACRYPT 2012: 18th International
  Conference on the Theory and Application of Cryptology and Information
  Security, Beijing, China, December 2-6, 2012. Proceedings 18}, pages
  626--643. Springer, 2012.

\bibitem{DBLP:conf/ccs/BougonCCDDDGT22}
Mikael Bougon, Herv{\'{e}} Chabanne, V{\'{e}}ronique Cortier, Alexandre Debant,
  Emmanuelle Dottax, Jannik Dreier, Pierrick Gaudry, and Mathieu Turuani.
\newblock Themis: An on-site voting system with systematic cast-as-intended
  verification and partial accountability.
\newblock In Heng Yin, Angelos Stavrou, Cas Cremers, and Elaine Shi, editors,
  {\em Proceedings of the 2022 {ACM} {SIGSAC} Conference on Computer and
  Communications Security, {CCS} 2022, Los Angeles, CA, USA, November 7-11,
  2022}, pages 397--410. {ACM}, 2022.

\bibitem{chaum2008scantegrity}
David Chaum, Richard Carback, Jeremy Clark, Aleksander Essex, Stefan
  Popoveniuc, Ronald~L Rivest, Peter~YA Ryan, Emily Shen, Alan~T Sherman,
  et~al.
\newblock Scantegrity ii: End-to-end verifiability for optical scan election
  systems using invisible ink confirmation codes.
\newblock {\em EVT}, 8(1):13, 2008.

\bibitem{clarkson2008civitas}
Michael~R Clarkson, Stephen Chong, and Andrew~C Myers.
\newblock Civitas: Toward a secure voting system.
\newblock In {\em 2008 IEEE Symposium on Security and Privacy (sp 2008)}, pages
  354--368. IEEE, 2008.

\bibitem{cortier2019beleniosvs}
V{\'e}ronique Cortier, Alicia Filipiak, and Joseph Lallemand.
\newblock Belenios{VS}: Secrecy and verifiability against a corrupted voting
  device.
\newblock In {\em 2019 IEEE 32nd computer security foundations symposium
  (CSF)}, pages 367--36714. IEEE, 2019.

\bibitem{culnane2015vvote}
Chris Culnane, Peter~YA Ryan, Steve Schneider, and Vanessa Teague.
\newblock v{V}ote: a verifiable voting system.
\newblock {\em ACM Transactions on Information and System Security (TISSEC)},
  18(1):1--30, 2015.

\bibitem{devillez2024verifiable}
Henri Devillez, Olivier Pereira, and Thomas Peters.
\newblock Verifiable and private vote-by-mail.
\newblock {\em Cryptology ePrint Archive}, 2024.

\bibitem{fuller2023adaptive}
Benjamin Fuller, Abigail Harrison, and Alexander Russell.
\newblock Adaptive risk-limiting comparison audits.
\newblock In {\em 2023 IEEE Symposium on Security and Privacy (SP)}, pages
  3314--3331. IEEE, 2023.

\bibitem{gjosteen2012norwegian}
Kristian Gj{\o}steen.
\newblock The norwegian internet voting protocol.
\newblock In {\em E-Voting and Identity: Third International Conference, VoteID
  2011, Tallinn, Estonia, September 28-30, 2011, Revised Selected Papers 3},
  pages 1--18. Springer, 2012.

\bibitem{heiberg2014verifiable}
Sven Heiberg and Jan Willemson.
\newblock Verifiable internet voting in estonia.
\newblock In {\em Electronic Voting: Verifying the Vote (EVOTE), 2014 6th
  International Conference on}, pages 1--8. IEEE, 2014.

\bibitem{DBLP:conf/ccs/HuberKKL00R022}
Nicolas Huber, Ralf K{\"{u}}sters, Toomas Krips, Julian Liedtke, Johannes
  M{\"{u}}ller, Daniel Rausch, Pascal Reisert, and Andreas Vogt.
\newblock Kryvos: Publicly tally-hiding verifiable e-voting.
\newblock In Heng Yin, Angelos Stavrou, Cas Cremers, and Elaine Shi, editors,
  {\em Proceedings of the 2022 {ACM} {SIGSAC} Conference on Computer and
  Communications Security, {CCS} 2022, Los Angeles, CA, USA, November 7-11,
  2022}, pages 1443--1457. {ACM}, 2022.

\bibitem{kortum2021voter}
Philip Kortum, Michael~D Byrne, and Julie Whitmore.
\newblock Voter verification of ballot marking device ballots is a two-part
  question: Can they? mostly, they can. do they? mostly, they don't.
\newblock {\em Election Law Journal: Rules, Politics, and Policy},
  20(3):243--253, 2021.

\bibitem{DBLP:conf/eurosp/KustersL00020}
Ralf K{\"{u}}sters, Julian Liedtke, Johannes M{\"{u}}ller, Daniel Rausch, and
  Andreas Vogt.
\newblock Ordinos: {A} verifiable tally-hiding e-voting system.
\newblock In {\em {IEEE} European Symposium on Security and Privacy, EuroS{\&}P
  2020, Genoa, Italy, September 7-11, 2020}, pages 216--235. {IEEE}, 2020.

\bibitem{lindeman2012gentle}
Mark Lindeman and Philip~B Stark.
\newblock A gentle introduction to risk-limiting audits.
\newblock {\em IEEE Security \& Privacy}, 10(5):42--49, 2012.

\bibitem{lindeman2012bravo}
Mark Lindeman, Philip~B Stark, and Vincent~S Yates.
\newblock Bravo: Ballot-polling risk-limiting audits to verify outcomes.
\newblock In {\em EVT/WOTE}, 2012.

\bibitem{ottoboni2018risk}
Kellie Ottoboni, Philip~B Stark, Mark Lindeman, and Neal McBurnett.
\newblock Risk-limiting audits by stratified union-intersection tests of
  elections (suite).
\newblock In {\em International Joint Conference on Electronic Voting}, pages
  174--188. Springer, 2018.
\newblock ArXiv: \url{https://arxiv.org/pdf/1809.04235.pdf}.

\bibitem{pereira2021individual}
Olivier Pereira.
\newblock Individual verifiability and revoting in the estonian internet voting
  system.
\newblock {\em Cryptology ePrint Archive}, 2021.
\newblock \url{https://eprint.iacr.org/2021/1098.pdf}.

\bibitem{rivest2008notion}
Ronald~L Rivest.
\newblock On the notion of ‘software independence’ in voting systems.
\newblock {\em Philosophical Transactions of the Royal Society A: Mathematical,
  Physical and Engineering Sciences}, 366(1881):3759--3767, 2008.
\newblock \url{https://people.csail.mit.edu/rivest/pubs/RW06.pdf}.

\bibitem{wombat}
Alon Rosen, Amnon Ta-shma, Ben Riva, and Yoni Ben-Nun.
\newblock Wombat voting system.
\newblock {\em https://wombat.factcenter.org}, 2011.

\bibitem{ryan2009pret}
Peter~YA Ryan, David Bismark, James Heather, Steve Schneider, and Zhe Xia.
\newblock Pr{\^e}t {\`a} voter: a voter-verifiable voting system.
\newblock {\em IEEE transactions on information forensics and security},
  4(4):662--673, 2009.

\bibitem{ryan2016selene}
Peter~YA Ryan, Peter~B R{\o}nne, and Vincenzo Iovino.
\newblock Selene: Voting with transparent verifiability and
  coercion-mitigation.
\newblock In {\em Financial Cryptography and Data Security: FC 2016
  International Workshops, BITCOIN, VOTING, and WAHC, Christ Church, Barbados,
  February 26, 2016, Revised Selected Papers 20}, pages 176--192. Springer,
  2016.

\bibitem{ryan2013pretty}
Peter~YA Ryan and Vanessa Teague.
\newblock Pretty good democracy.
\newblock In {\em Security Protocols XVII: 17th International Workshop,
  Cambridge, UK, April 1-3, 2009. Revised Selected Papers 17}, pages 111--130.
  Springer, 2013.

\bibitem{spertus2022sweeter}
Jacob~V Spertus and Philip~B Stark.
\newblock Sweeter than suite: Supermartingale stratified union-intersection
  tests of elections.
\newblock In {\em International Joint Conference on Electronic Voting}, pages
  106--121. Springer, 2022.
\newblock ArXiv: \url{https://arxiv.org/abs/2207.03379}.

\bibitem{stark2008conservative}
Philip~B Stark.
\newblock Conservative statistical post-election audits.
\newblock 2008.
\newblock
  \url{https://projecteuclid.org/journals/annals-of-applied-statistics/volume-2/issue-2/Conservative-statistical-post-election-audits/10.1214/08-AOAS161.pdf}.

\bibitem{stark2008sharper}
Philip~B Stark.
\newblock A sharper discrepancy measure for post-election audits.
\newblock 2008.
\newblock
  \url{https://projecteuclid.org/journals/annals-of-applied-statistics/volume-2/issue-3/A-sharper-discrepancy-measure-for-post-election-audits/10.1214/08-AOAS171.pdf}.

\bibitem{stark2010super}
Philip~B Stark.
\newblock $\{$Super-Simple$\}$ simultaneous
  $\{$Single-Ballot$\}$$\{$Risk-Limiting$\}$ audits.
\newblock In {\em 2010 Electronic Voting Technology Workshop/Workshop on
  Trustworthy Elections (EVT/WOTE 10)}, 2010.

\bibitem{stark2020sets}
Philip~B Stark.
\newblock Sets of half-average nulls generate risk-limiting audits: Shangrla.
\newblock In {\em International Conference on Financial Cryptography and Data
  Security}, pages 319--336. Springer, 2020.
\newblock ArXiv: \url{https://arxiv.org/abs/1911.10035}.

\bibitem{stark2023overstatement}
Philip~B Stark.
\newblock Overstatement-net-equivalent risk-limiting audit: Oneaudit.
\newblock In {\em Financial Cryptography and Data Security - Voting Workshop}.
  Springer LNCS 13953, 2023.
\newblock \url{https://arxiv.org/abs/2303.03335}.

\bibitem{terelius2010proofs}
Bj{\"o}rn Terelius and Douglas Wikstr{\"o}m.
\newblock Proofs of restricted shuffles.
\newblock In {\em Progress in Cryptology--AFRICACRYPT 2010: Third International
  Conference on Cryptology in Africa, Stellenbosch, South Africa, May 3-6,
  2010. Proceedings 3}, pages 100--113. Springer, 2010.

\bibitem{zagorski2013remotegrity}
Filip Zag{\'o}rski, Richard~T Carback, David Chaum, Jeremy Clark, Aleksander
  Essex, and Poorvi~L Vora.
\newblock Remotegrity: Design and use of an end-to-end verifiable remote voting
  system.
\newblock In {\em International Conference on Applied Cryptography and Network
  Security}, pages 441--457. Springer, 2013.

\end{thebibliography}
\appendix

\section{Optimizations for multiple contests}
Up to now, we have assumed that the election consists of a single contest, resulting in the vote vector $\mathbf{b}=b_{n-1}\|\cdots\| b_0$ containing one bit per candidate within that contest. Therefore, the encrypted vote vector $\mathbf{e}$ is computed as:$$\mathbf{e}=\enc(\mathbf{b})=(\enc(b_{n-1}),\cdots,\enc(b_0))$$ To expand our design to accommodate multiple contests while maintaining a common ballot style, we can consider $\mathbf{b}$ as a concatenation of multiple vote vectors, with each vote vector corresponding to a distinct contest. For example, in an election featuring two contests, we express $\mathbf{b}$ as $\mathbf{b_1}||\mathbf{b_0}$, where $\mathbf{b_0}$ and $\mathbf{b_1}$ represent the vote vectors for contests 0 and 1, respectively. The encrypted vote vector $\mathbf{e}$ is accordingly computed as$$\mathbf{e}=\enc(\mathbf{b})=(\enc(\mathbf{b_1}),\enc(\mathbf{b_0}))$$   In this configuration, the overall process remains largely unchanged compared to a single-contest election, with the exception that NIZK proofs and discrepancy assignments must be conducted separately for each contest on a ballot. 

Furthermore, to extend our design to accommodate elections with multiple ballot styles, it's required to store the ballot style in plaintext alongside each ballot. This is necessary for tallying the ballots within each contest. We must also ensure that the ballot style for each ballot undergoes the same mixnet process as the ballot itself and is presented in plaintext format before the RLA. Otherwise, if the digital ballot that is selected for the RLA lacks a corresponding paper ballot, it would be unclear which contest's RLA statistics should be updated. 

Consider $\mathbf{b} = \mathbf{b_{k-1}} || \cdots  || \mathbf{b_0}$, representing a set of $k$ distinct contests, where each $\mathbf{b_i}$ is an $n_i$ vote vector representing the $i^\text{th}$ contest. The total number of options in total is denoted as $n$, where $n = \sum_i n_i$. In the typical scenario, we would require $n$ encryption values to represent $\mathbf{e} = \enc(\mathbf{b})$. Then, we conduct the tally and the RLA process separately for each contest. While this approach offers flexibility, it also entails a higher number of encryption values per ballot, namely $n$, which can potentially complicate and slow down the mixnet process. Thus,  we provide some optimizations below.

Let $\mathbf{e}$ comprise $n$ ElGamal encryption values corresponding to an $n$-bit vote vector. To accelerate the mixnet process, following the tally phase, we apply the following $\text{COMP}$ transformation to $\mathbf{e} = (e_{n-1}, \cdots, e_0)$, with each $e_i = (\alpha_i, \beta_i)$, in order to compact it into a single ElGamal encryption value and speed up the mixing process:$$e'=\text{COMP}(\mathbf{e})=(\prod_i\alpha_i^{2^i} \bmod p,\prod_i\beta_i^{2^i} \bmod p)$$

Our assumption here is that $n$ is smaller than the size of the group $|q|$ in bits. For larger vote vectors, we must partition the vector $\mathbf{e}$ and construct separate ciphertexts for each segment. While this approach is highly  efficient, it lacks flexibility for RLA purposes, as outlined below.

Due to the way ElectionGuard encryption works, the decryption output is $K^\mathbf{b}$. For large $n$ (e.g., $n=100$), finding $\mathbf{b}$ via exhaustive search is infeasible. Normally, if such a ballot is selected for an RLA, $\mathbf{b}$  is derived from the paper, and  $K^{\mathbf{b}}$ is compared with $\mathbf{e}$. If these values match, the discrepancy for all contests on that ballot is 0.  However, if they differ, or if a digital ballot with $\Type=\text{``not matching''}$ is selected for the RLA, then discrepancy values for all contests on that ballot must be considered as the maximum possible value (e.g., +2 for a plurality voting).

\section{Protocol specification}
    \subsection{ElectionGuard details}\label{subsec:election-guard}
    To achieve a decentralized architecture, ElectionGuard relies on the collaboration of multiple entities known as 'guardians,' jointly responsible for protecting voters' privacy. Each guardian independently generates its own public-private key pair. These individual public keys are then combined to derive the election public key, crucial for encrypting all voters' selections. Subsequently, in the decryption phase, each guardian computes a verifiable partial decryption, resulting in a full verifiable decryption process.

    To address scenarios where certain guardians may be unavailable, a robust cryptographic mechanism is employed. Guardians cryptographically share their secret keys, ensuring that a pre-defined threshold quorum  guardians is required to achieve full decryption.
    See~\cite{benaloh-electionguard} for details of the key generation protocol.

    \protname{} follows the ElectionGuard 2.0 specification for primitives, as described in \cite{benaloh-electionguard}.
    We give a brief overview of the encryption scheme here.

    \subsection{Notations and building blocks} \label{subsec:notationAndBuildingBlocks}
    The selections of voter $i$, choosing from $m$ candidates, are represented as a binary vector $\mathbf{b_i}$ of length $m$, consisting of zeros for unselected options and ones for selections. In cases where it is clear and doesn't cause confusion, the index $i$ can be omitted for simplicity. For example, Alice is another way to represent the vote vector $\mathbf{b}=[1\;0\;0]$ in a contest with the following candidate list $[\text{Alice}\;\text{Bob}\;\text{Charlie}]$.

    The encryption of the vector $\mathbf{b}$ is denoted as $\mathbf{e}=\enc(\mathbf{b})$, wherein each entry $\mathbf{e}_{j}$ for $j=\{1,\cdots,m\}$ represents the encryption of the corresponding vote entry $\mathbf{b}_{j}$.

    \paragraph{Encryption scheme} In ElectionGuard, encryption of votes is performed using a variant exponential form of the ElGamal cryptosystem. The cryptosystem parameters $(p, q, g)$ are defined as follows: $p$ is a prime number equal to $2kq + 1$, where $q$ is also a prime number, and $g$ is a generator of the order $q$ subgroup of $\mathbb{Z}_p^*$. (The
    ElectionGuard specification requires particular values of $p$ and $q$.) To generate a public-private key pair, a private key $s\in\mathbb{Z}_q$ is randomly selected. The corresponding public key, $K$, is then computed as $K = g^s \mod p$ and made public. To encrypt a vote $b$,  a     nonce $\xi$ is selected randomly such that $0 \leq \xi < q$. Then, encryption of $b$, denoted by $\enc(b)$, is computed as
    $$\enc(b)=(g^\xi \bmod p,K^{b}\cdot K^\xi \bmod p)=(g^\xi \bmod p,K^{b+\xi} \bmod p).$$
    The entity possessing the corresponding secret key $s$ can decrypt $(\alpha, \beta)$ as $\beta/\alpha^s \mod p = K^b \mod p$. If $b$ is sufficiently small, its value can be derived from $K^b$ using an exhaustive search.
    This encryption scheme is additively homomorphic. Namely,$$\enc(b_1,\xi_1)\cdot \enc(b_2,\xi_2)=\enc(b_1+b_2,\xi_1+\xi_2)$$

    We can use the following equation to re-encrypt the ballot $b$ $$\enc(b,\xi_1)\cdot \enc(0,\xi_2)=\enc(b,\xi_1+\xi_2)$$

    \paragraph{Hash function}
    The function $H$ that is used in this whole section is the same as one specified in ElectionGuard, which is based on
    HMAC-SHA-256. The function $H$ takes two inputs, $B_0$ and $B_1$, where $B_0$ is 256 bit long and corresponds to the key in HMAC and $B_1$, which is of arbitrary length, is the actual input to the HMAC. These inputs are separated by a semicolon. If $B_1$ is comprised of multiple elements (e.g. $B_1=a||b||c$), then these elements are separated by commas( e.g. $H(B_0;a,b,c)$).  \\

    \paragraph{Zero knowledge proofs}
    We use the following ElectionGuard zero knowledge proofs, which are assumed to be paramaterized
    by the ElectionGuard group parameters $(p,q,g)$.

    \begin{itemize}
        \item Distributed proof of decryption correctness, used by the guardians to decrypt output ciphertexts.
        \item
        \textbf{$\mathsf{KnowDlog}(K)$}: a proof of knowledge of the discrete log of $K$ base $g$
        (that is, of $x$ s.t. $K = g^x \bmod p$).\\
        \item
        \textbf{$\mathsf{EqDlogs}(x,y,X,Y)$}: a proof of equality of discrete logs.\\
        \item
        \textbf{Proof that $(a,b)$ is an encryption of an integer in the range $0,\cdots,L$}\\
        \item
        \textbf{Proof that $(a,b)=(g^\xi,K^\xi)$ is an encryption of zero or one}\\
        \item
        \textbf{Proof that $(\alpha,\beta)=(g^\xi,KK^\xi)$ is an encryption of zero or one}\\
    \end{itemize}

    See~\cite{benaloh-electionguard} for details of these zero knowledge proofs. These are
    applied for our ``ZKP of proper ballot construction''.

    \subsection{Verification Summary} \label{subsec:verification}

    There are three separate verification stages, which can be verified by different people: cast-as-intended verification, which is mostly done by the voter, universal (BB) verification, which can be done by any member of the public, verification by observers at the \voteCollectionCenter\ that the data on the BB matches the ballot papers and the data included in the tally and the RLA matches the BB.

    Each of these is detailed below.

    \subsubsection{Cast-as-intended verification}
    \label{sec:Cast-as-intended verification}

    \begin{boxDwhitePlaceHere}{By the voter}
        1- Verify that the ballot paper matches the intended vote.
    \end{boxDwhitePlaceHere}

    \begin{boxDwhitePlaceHere}{By the voter or anyone else at the \remoteVotingCenter}
        2- Verify the digital signature on the sticker stuck onto the envelope.

        3- For ``BB Inclusion Check'': Verify that a properly-signed vote corresponding with the sticker's data is (eventually) on the BB,

        4- For ``Sticker BB Upload'': Verify that the sticker's data appears on the BB.
    \end{boxDwhitePlaceHere}

    \subsubsection{BB transcript verification (public)}
    \label{sec:BB-public-ver}

    \begin{enumerate}
        \item For each vote on the BB:
        \begin{enumerate}
            \item Verify voter's digital signature
            \label{item:dig-sig}
            \item Verify ZKP of proper ballot construction
            \label{item:zkp-proof}
        \end{enumerate}
        \item Verify that no two ballots on the bulletin board contain a common ciphertext (only necessary for privacy. See Section \ref{subsec:privacy}).
        \label{item:security-privacy}
        \item Verify mixing ZKP
        \label{item:zkp-mix}
        \item Verify decryption of all serial numbers $sn_i$.
        \label{item:sn-dec}
        \item Verify that each $sn_i$ is unique.
        \label{item:sn-unique}
        \item Verify aggregation of $\enc(\mathbf{Tally})$.
        \label{item:agg-tally}
        \item Verify the decryption $\mathbf{Tally}$ of $E(\mathbf{Tally})$.
        \label{item:zkp-tally-dec}
        \item For ``Sticker BB Upload'': verify all uploaded sticker digital signatures and verify the presence of a corresponding ballot on the BB for each sticker.
        \label{item:stickerupload}
    \end{enumerate}

    \subsubsection{Verification at the \voteCollectionCenter}
    \label{sec:Ver-at-voteCollectionCenter}

    \begin{enumerate}
        \item Verify digital signatures on the stickers, and verify that they  accurately sign data that corresponds to the information on the BB (i.e., with correct $Id_i$, $Id_e$ and $H_i$). (By both election authorities and auditors at the \voteCollectionCenter. Auditors can use their own electronic devices)
        \label{item:dig-sig-county}

        \item Verify that the \textit{Eligibility Problem Stack} includes all (and only) duplicate envelopes from the same voter.  
        \label{item:elig-stack}

        \item Verify that the \textit{Rejected Envelopes Stack} includes all (and only) envelopes with the invalid signatures.
        \label{item:reject-stack}

        \item Verify that $\mathbf{Tally}$ is properly added to the local CVRs or tallies.
        \label{item:tally-added-to-CVRs}
        \item Verify $\DecryptZKP{}{}$  for each ballot that is privately decrypted during the RLA process.
        \label{item:zkp-dec-selection}

        \item If ``not matching'' ballots are handled using the ONEAudit method (See Remark \ref{remark:oneaudit}), verify the aggregation of ``not matching'' ballots and its decryption.
        \label{item:dec-tally-n-m}
        \item Verify that discrepancies are properly determined and then correctly incorporated into the RLA.
        \label{item:disc-RLA-correctness}
        \item Verify there are no  ``serial number error''s.
        \label{item:no-sn-err}
       \item For the case where ballots are collected by a responsible authority (\i.e., \nodelaynostuff\ model), verify if the number of ballots in the ballot box for each ballot style does not exceed the corresponding        quantities on the BB.

        \label{box-num-leq-BB-num}
    \end{enumerate}

    \subsection{Failure handling}\label{subsec:failure-handling}
    In this section, we clarify the procedure that must be adopted if each aforementioned verification fails.
    \subsubsection{Cast-as-intended verification}
    \textbf{Item 1}: If the paper ballot does not match the voter's selections, this could be attributed to a voting machine failure or a mistake by the voter. In such cases, the voter has the option to either retry the voting process or abandon it.

    In either case, the following steps are taken:
    \begin{itemize}
        \item     An authority records the voter's identity and the associated sticker, which is considered invalid from that point onward. If an envelope with an invalidated sticker arrives at the \voteCollectionCenter, it is rejected. This recorded information serves as a reference in case the voter later claims possession of a sticker whose corresponding ballot is not found on the BB.
        \item The invalidated digital ballot is marked on the BB.

    \end{itemize}

    Furthermore, if the voter chooses to retry voting, the voting machine captures the new digital ballot as their vote and generates a corresponding paper ballot and sticker.

    \textbf{Item 2}: If the digital signature on the sticker is not verified, it is the voting machine's  fault.

    \textbf{Item 3}:  If a properly-signed vote, corresponding to the sticker's signature, from this voter does not (eventually) appear on the BB, it is the voting machine's  fault.
    \subsubsection{BB transcript verification (public)}
    \textbf{Items \ref{item:dig-sig} and \ref{item:zkp-proof}}: If for any digital ballot on the BB, digital signature or $\ValidZKP{}$ is not successfully verified,  it is the voting machine's  fault.

    \textbf{Item \ref{item:security-privacy}}: Given the negligible probability that any two ciphertexts are identical, if a ciphertext matches an already published one, it is overwhelmingly likely to be the voting machine's fault (though if
    the duplicated ciphertexts come from two different voting computers, it is not clear which one is at fault, unless
    there is reliable timing information).

    \textbf{Item  \ref{item:zkp-mix}}: If mixing ZKP is not successfully verified, it is the back-end server's fault. The RLA process must be performed after resolving the issue.

    \textbf{Item \ref{item:sn-dec}}: If decryption of each encrypted serial number  is not successfully verified, it is the back-end server's fault.   The RLA process must be performed after resolving the issue.

    \textbf{Item \ref{item:sn-unique}}: Given that it is of negligible probability that any two randomly generated serial numbers are equal, if each serial number is not unique, it is highly likely the voting machine's  fault.

    \textbf{Item  \ref{item:agg-tally}}: If computation of $\enc(\mathbf{Tally})$ is not successfully verified, it is the back-end server's fault. Decryption of $\enc(\mathbf{Tally})$ and then adding the decryption result to tallies from other systems must be performed after resolving the issue.

    \textbf{Item  \ref{item:zkp-tally-dec}}: If decryption of encrypted tally is not successfully verified, it is the back-end server's fault. Decryption of $\enc(\mathbf{Tally})$ and then adding the decryption result to tallies from other systems must be performed after resolving the issue.

    \textbf{Item  \ref{item:stickerupload}}:  For the case where ``Sticker BB Upload'' is used, if there is not a corresponding ballot for each uploaded sticker, it is the voting machine's  fault. 
    \subsubsection{Verification at the \voteCollectionCenter}
    \textbf{Item \ref{item:dig-sig-county}}:
    \begin{itemize}
        \item If the sticker on any envelope carries a valid digital signature from the voting machine, but the corresponding signed data does not appear on the BB, it is the voting machine's fault (unless the sticker has already been invalidated as part of the Cast-as-intended verification in item 1).
        \item If two stickers carry different valid signatures from a voting machine on the same message, that is the voting machine's fault.
        \item If two stickers carry valid signatures from different voting machines on the same voter' data, that is the voter's fault.
        \item If any two stickers contain  a same valid signature by a voting machine, that could be considered a fault of anyone with access to the sticker (e.g., the voting machine, the voter or even someone working in the postal system).
    \end{itemize}

    \textbf{Items \ref{item:elig-stack} and \ref{item:reject-stack}}: Any fault in the envelope stacking process is either due to the fault of the involved authorities or the fault of the device used for the signature verification.

    \textbf{Item \ref{item:tally-added-to-CVRs}}: Any fault is the involved authorities' fault.

    \textbf{Items \ref{item:zkp-dec-selection} and \ref{item:dec-tally-n-m}}: If decryption of any ballot or aggregation of any sets of ballots and then its decryption are not successfully verified, it is the back-end server's fault.

    \textbf{Item \ref{item:disc-RLA-correctness}}: If discrepancies are not correctly determined and then integrated into the RLA computations, it is the involved authorities' and the RLA calculator's fault, respectively.

    \textbf{Item \ref{item:no-sn-err}}: Any ``serial number error'' is the scanning device's fault.

    \textbf{Item \ref{box-num-leq-BB-num}}:  Any fault in counting the number of ballots in the ballot box is the involved authorities' fault.

    \section{Security Proofs}
    \subsection{Verifiability} \label{subsec:verifiability-proof}

We prove the Propositions required for \autoref{coro:main-result}. The real and ideal worlds are given in \autoref{subsec:security-verifiability}.
    This section provides the intermediate worlds and proofs of their equivalence that are used to prove the main verification result,
    Corollary~\ref{coro:main-result}.

    \paragraph{Intermediate World 1} Ballots in this world closely resemble those in the Ideal World. However, the RLA is carried out by randomly sampling from the digital records, identifying their corresponding paper ballots, and subsequently adjusting the RLA statistics in accordance with any discrepancies.

    \paragraph{Intermediate World 2} In this world, the ballots closely resemble those in Intermediate World 1. However, a unique serial number is recorded alongside the voter's unique identifier on digital ballots. Uniqueness of serial numbers published on the BB can be publicly verified. This serial number effectively replaces the voter's unique identifier on the corresponding paper ballot, serving as an alias for the voter. So, the authority in this world does not check if voters' identifiers printed on the paper ballots form a subset of the those recorded on the BB. The authority still checks if the quantity of votes in the ballot box is equivalent to the number of ballots recorded on the BB. The authority also ensures if the identities, published on the BB, are not spoofed. In this world, there is a potential for a malicious voting machine to deviate from the protocol by printing a permutation of serial numbers (which are published on the BB) on the paper ballots. In other words, although the set of all serial numbers published on the BB are the same as those printed on the paper ballots, the serial number that is assigned to voter $i$ on the BB, might be printed on voter $j$'s paper ballot where $i\neq j$. Consequently, the RLA process in this world unfolds as follows: A digital ballot is chosen at random for the RLA process, and its serial number is utilized to locate the corresponding paper ballot. Subsequently, discrepancies are accurately identified and appropriately integrated into the RLA statistics.

    \paragraph{Intermediate World 3} In this world, the ballots closely resemble those in Intermediate World 2. However, dishonest voters in this world might choose not to put their ballot in the ballot box. Also, the authority checks if the quantity of ballots in the box does not exceed the number of ballots on the BB (i.e., Section \ref{sec:Ver-at-voteCollectionCenter}, Verification at the \voteCollectionCenter, item \ref{box-num-leq-BB-num}). Authority still ensures if voters identities, recorded on the BB, are not
    spoofed. Moreover, a malicious voting machine might deviate from the protocol by printing arbitrary, non-matching serial numbers on paper ballots. Consequently, the RLA process in this world unfolds as follows: A digital ballot is randomly selected for RLA and subsequently categorized into one of the following groups, each with its corresponding discrepancy specification, where authorities at the \voteCollectionCenter\ verify this operation:

    \begin{itemize}
        \item `not matching': When the serial number of a digital ballot fails to match any serial number in the paper ballot manifest, the discrepancy for the selected ballot is assigned by employing a worst-case paper assumption for the given ballot.
        \item `one-to-one': When the serial number of a digital ballot matches a single entry in the paper ballot manifest, its discrepancy is computed accordingly and incorporated into the RLA statistics. It's important to note that there is a one-to-one correspondence between `one-to-one' digital ballots and paper ballots.
        \item `duplicated': When the serial number of a digital ballot matches multiple entries in the paper ballot manifest, the discrepancy between the digital record and each of those multiple entries is computed individually. Among these calculated discrepancies, the minimum discrepancy value is selected for inclusion in the RLA statistics. The paper ballot associated with this minimum discrepancy is identified as the corresponding paper ballot for the digital ballot. In cases where multiple paper ballots share the same minimum discrepancy, one of them is randomly chosen to serve as the corresponding paper ballot for the digital ballot.
    \end{itemize}

    \paragraph{Intermediate World 4 (only defined for \nostuff\ model)}  In this world, the voting process and the paper path differ from the previous world as follows:
    \begin{enumerate}
        \item There is a trusted party who honest voters can rely on to guarantee the presence of a ballot corresponding to their unique identifier on the BB,
        \item Voters place their ballot inside an envelope labeled with their ID and send it to the \voteCollectionCenter,
        \item Every paper ballot undergoes the same processing steps as those within \protname{}\ (steps 6 and 7 from the paper path explained in Section \ref{subsec:paperPath}). However, rather than checking the signatures on the stickers, jurisdiction authorities accept any envelope labeled with IDs on the BB, and
        \item The adversarial model is according to \nostuff\ model, defined  in Section \ref{sec:attacker model}.
    \end{enumerate}

    \paragraph{Intermediate World 5} The only distinction between this world and Intermediate World 4 (or Intermediate World 3 in the \nodelaynostuff\ model) is  that the scanning device might malfunction by recording incorrect serial numbers.

    \paragraph{Intermediate World 6} Differences between this world and the Intermediate World 5 are as follows: The serial numbers and the vote contents of digital ballots are initially recorded by the voting machine in an encrypted format. Each digital ballot is also accompanied by a zero-knowledge proof demonstrating the validity of the (encrypted) vote. Subsequently, all digital ballots undergo the same processing steps as those within \protname{}\ (i.e., steps \ref{step1-dig} to \ref{step:lastStepOfElectronicProcess} from the digital path explained in Section \ref{dig-path}) followed by private decryption in step 9.

    \paragraph{Intermediate World 7}
    In this world, instead of depending on a trusted third party, the voter's digital signature on their digital ballot serves as publicly verifiable evidence that the voter's identity is not spoofed.

    \begin{remark}
        Intermediate World 7 is the same as our protocol in the \nodelaynostuff\ model.
    \end{remark}

    Intermediate World $x$ game with $x\in\{1,\cdots,7\}$  is defined similarly to the Ideal World game (outlined in Section \ref{subsec:security-verifiability}), with the distinction that
    it is set in the Intermediate World $x$.

    \begin{prop}
        \label{prop:BMD sig}
        Given that
        \begin{itemize}

            \item successful verification of each digital signature on a sticker by its corresponding voter (i.e., Section \ref{sec:Cast-as-intended verification}, Cast-as-intended verification, item 2),
            \item no reported errors from voters regarding missing votes on the BB (i.e., Section \ref{sec:Cast-as-intended verification}, Cast-as-intended verification, item 3),
        \end{itemize}
        for any adversary $\mathcal{A}$ who wins the Real World Game in the \nostuff\ model, then there exists an adversary $\mathcal{A}'$ who wins the Intermediate World 7 Game.
    \end{prop}

    \begin{proof}
        The only distinction between the Real World in the \nostuff\ model and the Intermediate World 7 is the method used to prevent digital ballot drops. In the Intermediate World 7, a trusted party ensures that voters' ballots are present on the BB. In the Real World, this assurance is achieved through the voting computer's signature on the sticker. So to prove the proposition, we first demonstrate that the voting machine's digital signature $\auth_{\textit{Voting Computer}} (Id_i,Id_e,H_i)$, printed on the sticker which is provided to a voter with unique identifier $Id_i$, proves  the voting machine's commitment  to publish a digital ballot with the hash value $H_i$ for a voter with the same identifier $Id_i$ in an election with the same election identifier $Id_e$ on the BB. 
        
        Consider the opposite in the Real World: a scenario where $\auth_{\textit{Voting Computer}} (Id_i,Id_e,H_i)$ does not prove the voting machine's commitment to publish the corresponding digital ballot on the BB. The private key of the voting machine is securely held by the voting machine. If, however, anyone manages to generate a valid digital signature without the voting machine's involvement and hence without access to their private key, it would compromise the security of the digital signature system. Such a situation contradicts our initial assumption regarding the security of the cryptographic primitives employed in the system. Therefore, since honest voters verify the signature on the sticker and they report no errors, the ballots of honest parties are recorded on the BB. Also, due to having a valid sticker, their envelopes, once received, will be accepted at the \localVotingCenter. If the serial numbers and vote selections on the paper and digital ballots in the Intermediate World 7 match those in the Real World, then for any sequence of ballots selected for auditing, the number of discrepancies detected by the audit in the Real World will be equal to that in the Intermediate World 7. Consequently, the RLA output in both worlds would be the same.

    \end{proof}

    \begin{prop}
        \label{prop:exclusion}
        Given that all digital signatures on the BB are successfully verified  (i.e., Section \ref{sec:BB-public-ver}, BB transcript verification (public), item \ref{item:dig-sig}),  for any adversary $\mathcal{A}$ who wins the Intermediate World 7 Game, then there exists an adversary $\mathcal{A}'$ who wins the Intermediate World 6 Game.

    \end{prop}

    \begin{proof}
        The only distinction between the Intermediate World 7 and the Intermediate World 6 is the method used to prevent digital ballot stuffing. In the Intermediate World 6, a trusted party ensures that identities recorded on the BB are not spoofed. But in the Intermediate World 7, the voter's digital signature on their ballot achieves this goal. So, we first demonstrate that the voter's digital signature proves that their identity is not spoofed. 
        
        Consider the opposite in the Intermediate World 7: a scenario where a voter's identity on the BB is spoofed. The private key corresponding to each voter's digital signature is securely stored in their smart card, accessible only to the voter. If, however, the voting machine manages to generate a valid digital signature on the BB without the voter actively attempting to cast a vote and hence without access to their smart card, it would compromise the security of the digital signature system. Such a situation contradicts our initial assumption regarding the security of the cryptographic primitives employed in the system. Therefore, if all digital signatures on the BB are successfully verified, voters' identities are not spoofed. Moreover, if the serial numbers and vote selections on the paper and digital ballots in the Intermediate World 6 match those in the Intermediate World 7, then for any sequence of ballots selected for auditing, the number of discrepancies detected by the audit in the Intermediate World 7 will be equal to that in the Intermediate World 6. Consequently, the RLA output in both worlds would be the same.
    \end{proof}

    \begin{prop}
        \label{prop:enc-mix}
        Given that correctness of
        \begin{itemize}
            \item all $\ValidZKP{}$ proofs (i.e., Section \ref{sec:BB-public-ver}, BB transcript verification (public), item \ref{item:zkp-proof}),
            \item all $\DecryptZKP{}{}$ proofs including:
            \begin{itemize}
                \item decryption of all serial numbers (i.e., Section \ref{sec:BB-public-ver}, BB transcript verification (public), item \ref{item:sn-dec}),
                \item decryption of vote selections for ballots that are selected for the RLA process  (i.e., Section \ref{sec:Ver-at-voteCollectionCenter}, Verification at the \voteCollectionCenter, item \ref{item:zkp-dec-selection}), and
                \item decryption of the encrypted tally obtained at step \ref{step:digitalDecryptedTally} of the digital path (i.e., Section \ref{sec:BB-public-ver}, BB transcript verification (public), item \ref{item:zkp-tally-dec})
            \end{itemize}
            \item $\MixZKP{}{}$ proof (i.e., Section \ref{sec:BB-public-ver}, BB transcript verification (public), item \ref{item:zkp-mix}), and
            \item aggregation of encrypted ballots at step \ref{step:digitalEncryptedTally} of the digital path (i.e., Section \ref{sec:BB-public-ver}, BB transcript verification (public), item \ref{item:agg-tally})
        \end{itemize}
        are successfully verified, for any adversary $\mathcal{A}$ who wins the Intermediate World 6 Game, then there exists an adversary $\mathcal{A}'$ who wins the Intermediate World 5 Game.
    \end{prop}

    \begin{proof}

        Let $\mathcal{B}$ denote the list of digital ballots  recorded on the BB for $N$ voters before initiation of the step \ref{step1-dig} in the Intermediate World 6:
        $$\mathcal{B}=\{(\enc(sn_i),\enc(\mathbf{b_i}):i\in\{1,\cdots,N\}\}\}$$
        By relying on the security guarantees provided by the $\ValidZKP{}$ protocol, the validation of $\ValidZKP{}$ for all ballots ensures that, for any $i \in \{1, \cdots, N\}$, $\mathbf{b_i}$ represents a valid vote. This crucial property will later enable us to use the same vote selections in the Intermediate World 5. After applying the mixnet, the list transforms to $\mathcal{B}'$:
        $$\mathcal{B}'=\{(\enc(sn'_j),\enc(\mathbf{b'_j})):j=\pi(i),i\in\{1,\cdots,N'\}\}\}$$
        Due to validity of $\MixZKP{}{}$, we have $N'=N$, $sn'_j=sn_i$, $\mathbf{b'_j}=\mathbf{b_i}$ for any $i\in\{1,\cdots,N\}$ with $\pi$ denoting a permutation from $\{1,\cdots,N\}$ to $\{1,\cdots,N\}$. Deviation from these conditions would violate the mixnet security assumptions. Some digital ballots in $\mathcal{B}'$ are then privately decrypted and used for the RLA process. Let digital ballots after private decryption be denoted by $\mathcal{B}''$:$$\mathcal{B}''=\{(sn''_j,\mathbf{b''_j}):j=\pi(i),i\in\{1,\cdots,N\}\}\}$$
        The validity of $\DecryptZKP{}{}$ ensures $sn''_j = sn'_j$ and $\mathbf{b''_j} = \mathbf{b'_j}$ (and hence $sn''_j = sn_i$ and $\mathbf{b''_j} = \mathbf{b_i}$), with $j = \pi(i)$. Otherwise, the security of the proof of decryption protocol would be violated.

        Now let  $$\mathcal{B}'''=\{(sn_i,\mathbf{b_i}):i\in\{1,\cdots,N\}\}\}$$
        be the digital ballots recorded by the voting machine on the BB in Intermediate World 5. The tally in Intermediate World 5 ($\sum_{i=1}^{N}\mathbf{b_i}$) matches the result of decryption at step \ref{step:digitalEncryptedTally} of the protocol in the Intermediate World 6 ($\dec(\prod_{j=1}^{N}\enc(\mathbf{b'_j}))$). This equivalence is due to the homomorphic properties of the encryption algorithm. The private decryption at step \ref{step:digitalDecryptedTally} of the protocol in the Intermediate World 6 results in $\sum_{j=1}^{N}\mathbf{b'_j}$. Otherwise the security of the proof of decryption protocol would be violated. Since the paper path in both worlds is the same, the paper manifest in Intermediate World 5 can be selected to match that in Intermediate World 6. Each sequence of ballots in Intermediate World 6 corresponds to a sequence in Intermediate World 5 (established through the $\pi^{-1}$ permutation function), maintaining identical discrepancy values. Consistency in both discrepancy values and tally across these intermediate worlds ensures an equivalent RLA risk.
    \end{proof}

    \begin{prop}
        \label{prop:mal-scan}
        Given that no ``serial number error'' is observed during the RLA (i.e., Section \ref{sec:Ver-at-voteCollectionCenter}, Verification at the \voteCollectionCenter, item \ref{item:no-sn-err}),  for any adversary $\mathcal{A}$ who wins the Intermediate World 5 Game,  there exists an adversary $\mathcal{A}'$ who wins the Intermediate World 4 Game (or Intermediate World 3 Game for the \nodelaynostuff\ model).
    \end{prop}

    \begin{proof}

        In Intermediate World 5, when a digital ballot is chosen for the RLA, its discrepancy is  set to the maximum possible value for that ballot unless a corresponding paper ballot with the same serial number and a lower discrepancy is identified in the paper record. If such a paper ballot exists in Intermediate World 5, it implies the coexistence of the same digital and paper ballot in Intermediate World 4. Consequently, the discrepancy for a digital ballot in Intermediate World 4 will not be less than its corresponding discrepancy value in Intermediate World 5.

    \end{proof}

    \begin{prop}
        \label{prop:paper-path-4-to-3}
        Given  that \textit{Eligibility Problem} Stack includes all (and only) duplicate envelopes from the same voter (i.e., Section \ref{sec:Ver-at-voteCollectionCenter}, Verification at the \voteCollectionCenter, item \ref{item:elig-stack}), \textit{Rejected Envelopes} Stack includes all (and only) rejected envelopes (i.e., Section \ref{sec:Ver-at-voteCollectionCenter}, Verification at the \voteCollectionCenter, item \ref{item:reject-stack}), for any adversary $\mathcal{A}$ who wins the Intermediate World 4 Game, then there exists an adversary $\mathcal{A}'$ who wins the Intermediate World 3 Game.
    \end{prop}

    \begin{proof}

        Let's start with the setup in the Intermediate World 4. Let $N$ denote the total digital ballots on the BB. Assume the number of honest voters be denoted by $n$. Since there is a trusted party who honest voters can rely on to guarantee the presence of their ballots on the BB, $n$ digital records on the BB correspond to these honest voters. So, we have $N\geq n$.  In  the \nostuff\ model, although the envelopes mailed by honest voters might drop, the adversary is unable to stuff new envelopes on behalf of honest voters. Furthermore, the envelope that each honest voter mails, once received,  remains unaltered.  Therefore, paper ballots inside envelopes received from honest voters at the \voteCollectionCenter\ in the Intermediate World 4, represented by $\mathcal{P}'_h$, form a subset of paper ballots viewed and mailed by the honest voters, represented by $\mathcal{P}_h$. Similarly, $\mathcal{P}_{\bar{h}}$ and $\mathcal{P}'_{\bar{h}}$ denote the paper ballots mailed by dishonest voters and received at the \voteCollectionCenter, respectively.

        Let the number of voters and the honest voters in the Intermediate World 3 be the same as in the Intermediate World 4 (i.e., $N$ and $n$, respectively) with the same digital records  recorded on the BB. Assume that the paper ballots provided to the voters in the Intermediate World 3 are identical to those in the Intermediate World 4 with one distinction: The serial numbers printed on $\mathcal{P}_h\backslash\mathcal{P}'_h$ and $\mathcal{P}_{\bar{h}}\backslash\mathcal{P}'_{\bar{h}}$ do not match any serial number on the BB. Figure \ref{fig:proof-4-3} demonstrates these setups.

        \begin{figure}
            \begin{center}

                \includegraphics[scale=0.35]{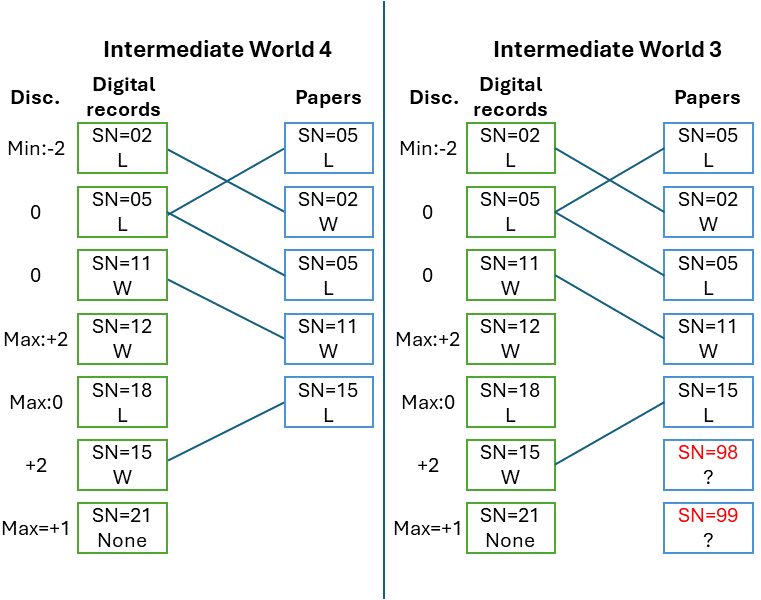}
                \caption{Building ballots in Intermediate World 3 based on those in Intermediate World 4.}
                \label{fig:proof-4-3}
            \end{center}

        \end{figure}

        Then, these paper ballots do not match any records on the BB, and therefore, their presence in Intermediate World 3 does not affect any discrepancy values. Consequently, the discrepancy of each digital record in Intermediate World 3 will be equal to its corresponding ballot in Intermediate World 4. This shows for any adversary $\mathcal{A}$ who wins the Intermediate World 4 Game, there exists an adversary $\mathcal{A}'$ who wins the Intermediate World 3 Game

    \end{proof}

    \begin{prop}
        \label{prop:paper-path-3-to-2}
        For any adversary $\mathcal{A}$ who wins the Intermediate World 3 Game, then there exists an adversary $\mathcal{A}'$ who wins the Intermediate World 2 Game.
    \end{prop}

    \begin{proof}

        Let's start with the setup in Intermediate World 3. In this world, each digital ballot is categorized as either `not matching', `one-to-one', or `duplicated'. We will prove that, in each of these categories, each ballot's discrepancy is maintained or decreased in the transition to Intermediate World 2. Figure~\ref{fig:proof-3-2} shows an example
        of the transition.

        For each digital ballot in the `duplicated' category, one of the paper ballots with a matching serial number that results in the minimum discrepancy is identified as the corresponding paper ballot for that digital ballot, and
        this correspondence is retained in Intermediate World 2.
        Remember that ``minimum discrepancy'' means ``the minimum discrepancy number, where 0 is taken to be a match, overstatements are positive, and understatements are negative.'' Because we were choosing the ``minimum discrepancy'' for our RLA accounting in Intermediate World 3, it's obvious that these one-to-one matches we just made in Intermediate World 2 produces exactly the same discrepancies for every digital ballot in the `duplicated' category.

        In the example in Figure~\ref{fig:proof-3-2}, we have a duplicate for $\textit{SN} = 02.$ The minimum
        discrepancy isn't the one that exactly matches, it's the 2-vote understatement (a ``-2 discrepancy''). So in
        Intermediate World 2 we
        choose to keep the connection to the ballot with $(\textit{SN=02}, W)$ (which therefore gets marked with a blue line and retained into Intermediate World 2). This therefore has the same discrepany in Intermediate worlds~2 and~3.

        For digital ballots in the `one-to-one' category in Intermediate World 3, there is no change in Intermediate
        World 2. In Figure~\ref{fig:proof-3-2}, Serial numbers 05, 11 and 15 are in this category.

        Finally, consider how to deal with digital ballots that are in the `not matching' category in Intermediate World
        3. We want to make a one-to-one correspondence with the ballot papers, of which we must (by assumption) have no more than we have electronic records. So we just make one up! We pick an arbitrary assignment of as-yet-not-matched digital records to as-yet-not-matched papers. Paper serial numbers are edited if the paper's
        serial number matched no digital record, or if it was a duplicate that was not selected as the minimum when dealing
        with the `duplicated' category.  Papers are added to make the number of digital and paper records equal, each
        with a serial number to match an as-yet-unmatched digital ballot.\footnote{This is just a proof technique to demonstrate that the treatment of ballots is conservative. There
        is no actual alteration of serial numbers on ballots, nor addition of new ballot papers during the protocol.}

        In Figure~\ref{fig:proof-3-2}, this arbitrary assignment is performed for $\textit{SN} = 12$ (which is matched to
        the unused duplicate of $\textit{SN}=02$), for $\textit{SN} = 18$ (which is matched to the previously-unmatched
        $\textit{SN} = 31$) and for $\textit{SN} = 21$ (which is matched to an invented ballot with an arbitrary vote selection). In each case, the serial number of the paper record is updated to match is corresponding digital
        one.

        We need to argue that the RLA discrepancies must be no greater for these ballots, but that's easy because when we had no match (in Intermediate World 3) we assumed the worst-case discrepancy. Now we have a specific match for each of them, which can't, by definition, be worse than the worst case. (edited)

        So, overall, this process produces a one-to-one match between digital and paper ballots, in which no digital record's discrepancy has increased from Intermediate World 3 to Intermediate World 2.

        be denoted by $\mathcal{B}$ and $\mathcal{B}'$, respectively. Let the set of paper ballots and those corresponding to $\mathcal{B}'$ be denoted by $\mathcal{P}$ and $\mathcal{P}'$, respectively. The discrepancy for each ballot in $\mathcal{B}\backslash\mathcal{B}'$ is set to the maximum possible value for that ballot.

        Now, let the number of voters in Intermediate World 2 be identical to that in Intermediate World 3. For each voter, the ballot selections on the paper and on the BB, and the serial number recorded on the BB in Intermediate World 2 are selected to be identical to those in Intermediate World 3. Additionally, the serial number on the paper ballot for each voter whose paper ballot in Intermediate World 3 is in $\mathcal{P}'$ would be the same  in Intermediate World 2. However, the serial numbers on the paper ballots for other voters in Intermediate World 2 would be any arbitrary permutation of the serial numbers recorded on the BB for $\mathcal{B}\backslash\mathcal{B}'$ in Intermediate World 3. Consequently, discrepancies for ballots in $\mathcal{B}'$ would be identical in both worlds. But the discrepancy for each ballot in $\mathcal{B}\backslash\mathcal{B}'$ in Intermediate World 2 would be equal to or less than that in Intermediate World 3. Therefore, the RLA risk in Intermediate World 2 would be equal to or greater than in Intermediate World 3. Figure \ref{fig:proof-3-2} demonstrates these setups.
        \begin{figure}
            \begin{center}

                \includegraphics[scale=0.35]{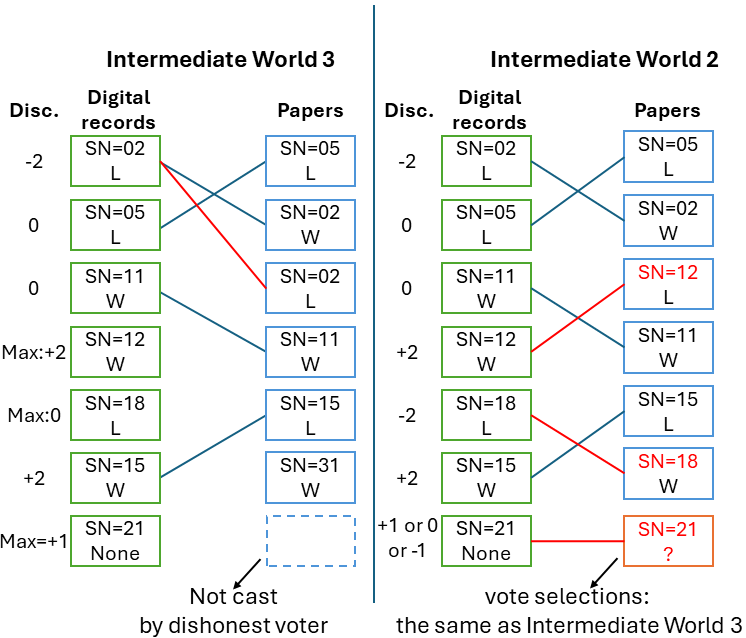}
                \caption{Building ballots in Intermediate World 2 based on those in Intermediate World 3.}
                \label{fig:proof-3-2}
            \end{center}
        \end{figure}

    \end{proof}

    \begin{prop}
        \label{prop:permutation}
        Given successful verification of the uniqueness of all serial numbers on the BB (i.e., Section \ref{sec:BB-public-ver}, BB transcript verification (public), item \ref{item:sn-unique}), for any adversary $\mathcal{A}$ who wins the Intermediate World 2 Game, then there exists an adversary $\mathcal{A}'$ who wins the Intermediate World 1 Game.
    \end{prop}

    \begin{proof}
        Let the permutations that is applied to serial numbers before printing them on the paper ballots in the Intermediate World 2 be denoted by $\pi$. Let the paper ballot and the digital ballot for the voter $i$ in the Intermediate World 2 be denoted by $P_i$ and $B_i$, respectively. Now assume the Intermediate World 1 is established as follows. Although the paper ballot that is printed for the voter $i$ is still $P_i$, the digital ballot that is recorded on the BB for the same voter would be $B_j$ with $j=\pi^{-1}(i)$. This way discrepancies of all ballots in both worlds would be identical. Figure \ref{fig:proof-2-1} demonstrates these setups.

        \begin{figure}
            \begin{center}

                \includegraphics[scale=0.35]{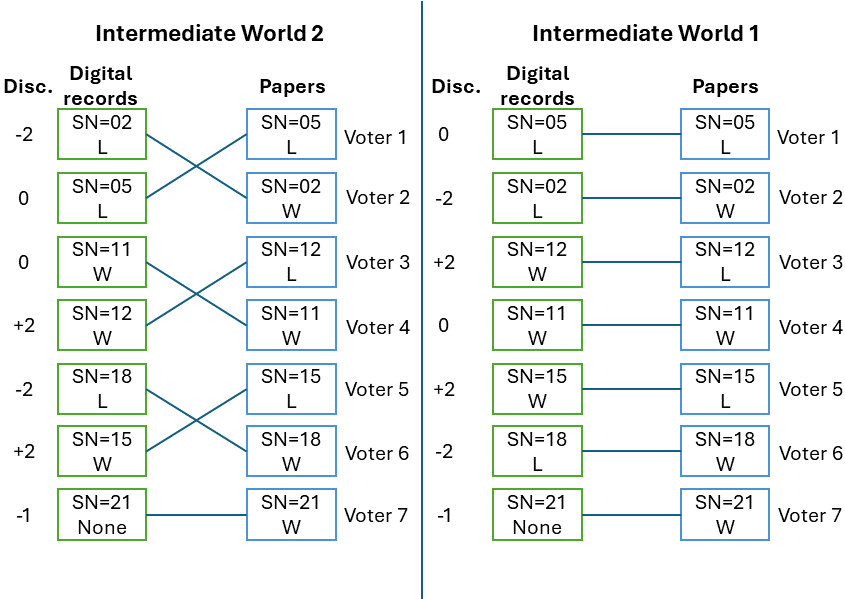}
                \caption{Building ballots in Intermediate World 2 based on those in Intermediate World 3.}
                \label{fig:proof-2-1}
            \end{center}
        \end{figure}

    \end{proof}

    \begin{prop}
        \label{prop:random-dig}
        For any adversary $\mathcal{A}$ who wins the Intermediate World 1 Game, then there exists an adversary $\mathcal{A}'$ who wins the Ideal World Game.
    \end{prop}
    \begin{proof}
        Due to the one-to-one correspondence between digital records and paper ballots (resulting from inclusion and exclusion criteria), and the fact that RLA involves the random selection of ballots, any sequence of paper ballots in the Ideal World corresponds to a matching sequence of digital ballots in the Intermediate World 1 with the same discrepancy values and hence with the same RLA risk.
    \end{proof}

    \subsection{Privacy} \label{subsec:privacy-proof}

    This section sketches a proof of \autoref{prop:main-privacy}. We begin by restating the privacy definition BRPIV from~\cite{bernhard2015sok}. We then list two important technical assumptions, which summarise the privacy provided by components. We then restate the proposition and provide the proof.

    The setup assumes that all votes are opened for audit, which is the worst case. In practice, most typical runs of \protname{} would be unlikely to open very many, if any, of the votes, thus revealing much less information than is assumed here.

    \begin{definition}[BPRIV, Adapted from \cite{bernhard2015sok} Definition~7]
        Consider a voting scheme $\mathcal{V} = (\textsf{Setup, Vote, Valid, Publish, Tally, Verify})$ for a set $I$ of voter identities. The scheme has \emph{ballot privacy} if there exists an algorithm $\textsf{SimProof}$ such that no efficient adversary $\mathcal{A}$ can guess a randomly-chosen bit $\beta$ with non-negligible advantage after playing the game $\textsf{Exp}_{\mathcal{A},\mathcal{V}}^{\textsc{BPRIV},\beta}(\lambda)$
        defined below, where $\lambda$ is the security parameter.

        \begin{itemize}
            \item \textbf{Setup.} The challenger picks a $\beta\in\{0,1\}$ uniformly at random. He sets up two empty bulletin boards $BB_0$ and $BB_1$ and two empty ballot boxes $\text{Box}_0$ and $\text{Box}_1$, creates a public key pair $(pk,sk)$ and posts the public information $pk$ on both boards. $\adv$ is given access to either $BB_0$ and $\text{Box}_0$ if $\beta = 0$ or to $BB_1$ and $\text{Box}_1$ otherwise.
            \item \textbf{Voting.} $\adv$ may make two types of queries.
            \begin{itemize}
                \item ($\mathcal{O}voteLR$) These model honest voters. $\adv$ provides a voter identity $id$ and two votes $(\mathbf{b_0}, \mathbf{b_1})$. The challenger generates a random serial number and creates a digital and paper ballot using $\mathbf{b_0}$. Similarly, the challenger generates another random serial number and creates a digital and paper ballot using $\mathbf{b_1}$. The challenger posts the digital ballot created from $\mathbf{b_0}$ and $\mathbf{b_1}$ on $BB_0$ and $BB_1$, respectively. She also drops the paper ballots created from $\mathbf{b_0}$ and $\mathbf{b_1}$ in $\text{Box}_0$ and $\text{Box}_1$, respectively.
                \item ($\mathcal{O}cast$) These are queries made on behalf of corrupt parties. $\adv$ provides a vote $\mathbf{b}$. The challenger creates a random serial number and then creates a ballot and posts it on both $BB_0$ and $BB_1$. $\adv$ is provided with two identical corresponding paper ballots and places either both or neither into their respective ballot boxes.
            \end{itemize}
            \item \textbf{Tallying,} $\mathcal{O}\textsf{Tally}()$. If $\beta=0,$ return valid tally $(r,\Pi) \leftarrow \mathsf{Tally}(BB_0,sk)$. If $\beta=1$, set $(r,\Pi) \leftarrow \textsf{Tally}(BB_0, sk)$
            and return simulated tally $(r, \textsf{SimProof}(BB_1, r))$. This includes showing $\adv$ the stickers and paper ballots from $\text{Box}_0$.
        \end{itemize}
    \end{definition}

\subsubsection{Assumptions}
    \protname{} relies on two critical technical components: the ElectionGuard protocol and the Verificatum mixnet. Proving that either of these have their required privacy properties is outside the scope of this paper. Instead, we state what we assume from them.

   The first assumption is that the basic ElectionGuard encryption scheme, including its validity ZKPs, is non-malleable.
    This applies to the encryption of an individual ciphertext. We first recall the definition of non-malleability (NM-CPA).

    \begin{definition}[NM-CPA, from \cite{bernhard2015sok}]
        NM-CPA is defined by the following game:
        \begin{enumerate}
            \item The challenger generates keys $(sk,pk)$ and gives the public key $pk$ to the adversary.
            \item The adversary picks two messages $m_0$ and $m_1$ and hands them to the challenger, who chooses a random bit
            $\beta$ and returns an encryption $c^*$ of $m_{\beta}$.
            \item The adversary may submit a vector $\vec{c}$ of query ciphertexts. For each query ciphertext $\vec{c}_i$, if $\vec{c}_i = c^*$, then the challenger returns $\bot$ (reuse of the challenge ciphertext is disallowed), otherwise he returns $\textsf{Dec}(sk, \vec{c}_i)$.
            \item The adversary submits a guess $\beta'$ for $\beta$.
        \end{enumerate}
        The adversary's advantage is $|\Pr(\beta' = \beta) - 1/2|$. We say the encryption scheme is NM-CPA secure if this advantage is negligible in the security parameter.
    \end{definition}

    This formalises the idea that the adversary cannot gain information about honest voters from reusing (possibly-manipulated versions of) their votes and thus inferring an honest vote from the tally. Together with the removal of duplicate ciphertexts, it implies that the votes of corrupted voters cannot depend on the votes of honest ones.

    \begin{assumption} \label{assumption:EG-NonMalleable}
    ElectionGuard's encryption scheme, including the validity ZKPs, is non-malleable (NM-CPA).
    \end{assumption}
    As far as we know there is no proof of this, but non-malleability (NM-CPA) is proven for Helios's almost-identical encryption-with-ZKP scheme in~\cite{bernhard2012not}.

    \begin{assumption}[\cite{terelius2010proofs}] \label{assumption:mix-zk}
    The Verificatum proof is an honest verifier zero knowledge proof of knowledge of the relation $\mathcal{R}_{\phi_{pk}}
    \lor \mathcal{R}_{\phi_{com}}$, where $\mathcal{R}_{\phi_{pk}}$ is the relation of valid shuffles and $\mathcal{R}_{\phi_{com}}$ is the relation that exposes the private commitment information. (See \cite{terelius2010proofs} for a precise statement.)
    \end{assumption}
    Informally, this proof says ``either I have performed a valid shuffle, or I have broken the commitment scheme.'' Assuming that the commitment scheme is sound, it implies a valid shuffle.

    \begin{assumption} \label{assumption:mix-physical}
        $\adv$ cannot distinguish the following two situations:
        \begin{itemize}
            \item The $\beta=0$ case: ballots arrive in signed envelopes and are physically shuffled, then the shuffled ballots are shown to the adversary.
            \item A different set of ballot papers arrive in the same signed envelopes, then a simulated shuffle is performed, with the same ballots from the $\beta=0$ case shown to the adversary.
        \end{itemize}
    \end{assumption}

    The proof of privacy for \protname{} employs a construction of \textsf{SimProof}  that closely follows the proof in \cite{bernhard2015sok} for Helios, and adds a simulation of the shuffle proof guaranteed by Assumption~\ref{assumption:mix-zk}.

\subsubsection{The proof}
    We are now ready to prove \autoref{prop:main-privacy}, which is restated here.

    \textit{
        \protname{} in the \nodelaynostuff\ model provides ballot privacy according to BPRIV~\cite{bernhard2015sok}, assuming that the underlying ElectionGuard ciphertexts satisfy NM-CPA and the mixing proof is zero knowledge.}

    The definitions of the voting scheme $\mathcal{V} = (\textsf{Setup, Vote, Valid, Publish, Tally, Verify})$  are as follows.
    \textsf{Setup} Sets up the election public key and PKI for signatures.

    \textsf{Vote$(id,v_0)$} Uses the election public key to encrypt the vote along with a randomly generated SN; signs it with the key associated with id (as specified in Voting step~\ref{step:make-vote}).

    \textsf{Valid$(BB_\beta, b_\beta)$} Verifies the vote validity ZKPs and digital signatures for each vote, then removes any
    vote that shares a ciphertext with any other (as specified in Digital Processing Step~\ref{step1-dig}).

    \textsf{Publish$(BB_\beta)$} The challenger publishes the encrypted input votes (from either $BB_0$ or $BB_1$), and the corresponding ciphertexts output from the mix (under whatever random permutation $\pi_0$ or $\pi_1$ respectively the mix generates). Also the aggregated encrypted input votes (though this can be derived).

    \textsf{Tally$(BB_0, sk)$} Returns the decrypted ballots in order from $BB$, accompanied by $\MixZKP{}{}$ and the decryption ZKP $\DecryptZKP{}{}$ for each vote.

    According to BPRIV, if $\beta=0$ the challenger directly posts the honestly-generated proof transcript $\mathsf{Tally(BB_0, sk)}$; if $\beta=1$ it instead posts the result $r$ from $\mathsf{Tally(BB_0, sk)}$ (which in our protocol consists of the aggregated tally and the permuted, decrypted ciphertexts) along with a simulation $\mathsf{Simproof(BB_1, r)}$ of the rest of the transcript. We need to specify $\mathsf{Simproof}$ so that this simulated transcript is indistinguishable from the honestly-generated transcript from the $\beta=0$ case.

    $\mathsf{Simproof(BB_1,r)}$ takes the true aggregated vote ciphertext from $BB_1$, the true mixing proof (input ciphertexts, output ciphertexts and $\MixZKP{}{}$) from $BB_1$, and the decrypted votes from $r$, and uses the ZK simulator of the decryption ZKP to simulate both the decryptions of the aggregate (from Step~\ref{step:digitalDecryptedTally}) and the individual vote decryptions (from Step~\ref{step:lastStepOfElectronicProcess}).

    \textbf{RLA.} The adversary observes the RLA procedure run on paper ballots corresponding to the left board.

    We model an adversary who colludes with sub-threshold number of guardians as having no knowledge of the private key.

    We can now prove privacy in the interactive honest-verifier zero knowledge setting---the extension to the noninteractive setting can follow standard techniques used in~\cite{bernhard2015sok}.
    \paragraph{Proof of \autoref{prop:main-privacy}}
    \begin{proof}
        We need to prove that using SimProof for simulating the tally above (in the case $\beta=1$) is indistinguishable from the real transcript (for the case $\beta=0$). The proof is very similar to the proof of BPRIV for Helios in \cite{bernhard2015sok}, except that the initial step simulates the mixing ZKP.

    \paragraph{[Game $G_{-1}$]} is the $\beta=0$ case, where the challenger runs an honest decryption and tally of $BB_0$ and opens ballot box $\text{Box}_0.$

    \paragraph{[Game $G_0$]} is the same as $G_{-1}$ except that all the ZKPs and the physical shuffle are simulated. The permuted list of decrypted output ballots, and the paper ballots shown to the adversary after shuffling, remain the same.
    \begin{itemize}
        \item Use Assumption~\ref{assumption:mix-zk} to produce a simulated shuffle ZKP indistinguishable from the real one.
        \item Use the Zero Knowledge property of ElectionGuard's decryption ZKPs to simulate them for all mixed votes.
        \item Use Assumption~\ref{assumption:mix-physical} to say that the physical shuffle can be ``simulated'' so that $\adv$ cannot distinguish the real shuffle from the one in which votes were substituted.
        \end{itemize}

    By the Zero knowledge property, this is indistinguishable from Game $G_{-1}$.

    The proof is now very similar to the proof of the equivalent step in~\cite{bernhard2015sok}.
       Let $n$ be the number of votes on $BB_0$. For $i=1 \ldots n$, define a series of games as follows.

        \paragraph{[Game $G_i$]} is derived from $G_0$ by replacing the first $i$ ballots on $BB_0$ with the corresponding ballots from $BB_1$.

        We now need to prove that, for $i=1,\ldots,n$ game $G_{i-1}$ is indistinguishable from game $G_i$, from which it differs only in the substitution of ballot $b_i$ on the BB.

        Consider what the adversary could have used for $b_i$ on board $BB_1$. If it is identical to $b_i$ from $BB_0$ then distinguishing is impossible. If it is different, then either the adversary re-uses  a ciphertext already on the board (in which case, the removal of duplicate ciphertexts corresponds to the return of $\bot$ in the NM-CPA game, no information is returned, and distinguishing is impossible), or the game is run and distinguishing the two boards---which differ only in one ballot---is equivalent to distinguishing those two ballots in the NM-CPA game. None of these options give $\adv$ a better advantage than distinguishing the two ballots in the NM-CPA game, which we assume to be negligible (Assumption~\ref{assumption:EG-NonMalleable}).

    We rely on Assumption~\ref{assumption:mix-physical} to argue that the paper evidence does not help the adversary to distinguish these two games.

    The last of these games ($G_n$) is exactly the $\beta=1$ case, so we have shown that it is indistinguishable from the $\beta=0$ game.
    \end{proof}

    \section{Using \protname{} with other RLA styles and avoiding the publication of \protname{} sub-tallies}
\label{App:otherRLAStyles}

\subsection{VAULT-\protname{}: incorporating \protname{} in to an RLA that uses VAULT for all ballots} \label{subsec:MergeWithVault}

    The VAULT RLA scheme~\cite{benaloh2019vault} protects against pattern-based coercion attacks (often called ``Italian attacks'') by hiding individual ballots unless they are selected for RLA. This protects against both the problem of small sub-tallies for \protname{}\ ballots and also the individual pattern-based coercion attacks on all ballots
    (at least, all ballots that are not selected for audit).

    There are two main steps to incorporate \protname{} into VAULT. In the first step, all the votes are tabulated on a VAULT BB. This step is simple: the \protname\ ballots that are \emph{output} by the mix are in exactly the right form for VAULT. The ordinary ballots from the \localVotingCenter{}\ can simply be appended. The tallies are then computed over the whole set of votes.

    In the second step, the RLA is conducted using the ballot-level comparison method, with samples taken at random from the VAULT-tabulated ballots. Discrepancies are calculated as follows:

    \begin{itemize}
        \item if the ciphertext is a \protname{}\ ballot, it should be treated as described in \autoref{discrepancy-det}, using a ballot paper of matching serial number if possible, and a worst-case assumption if there is no ballot paper,
        \item if the ciphertext is an ordinary ballot, the paper ballot should be found and the discrepancy calculated as in a standard VAULT RLA.\footnote{The only technical restriction, compared with standard VAULT, is that for the \protname{}\ ballots, the ciphertext must be decrypted rather than opened, because the authorities do not know the random factors used to generate it. The ordinary votes need to use the same encryption scheme as the \protname{}\ ones, so that they can be aggregated together---they may be either decrypted like the \protname{}\ ballots or opened (from the randomness used to generate them) like standard VAULT ballots.}
    \end{itemize}
    This method allows for a very efficient audit (ballot-level comparison for all ballots) with good privacy properties, even if the \protname{}\ ballots are only a small set.

    This would also work well if, for some reasons, the local electoral authority included some but not all ballots into a VAULT audit.

    \subsection{Sub-\protname: hiding \protname{} ballots in larger batches} \label{subsec:SubMerge}

    In this variant, serial numbers are not used, because paper ballots are not individually matched to their corresponding digital record. If it is known in advance that a certain ballot type will be processed using Sub-\protname{}, the serial numbers can be omitted from ballots of that type. This has a significant advantage for defending against a coercer who colludes with people at the \localVotingCenter{}, because even a voter who voluntarily reports their serial number (for example, by taking a photo) cannot be linked to their ballot.
    However, their anonymity set, among ballots of the same type, may still be small against a coercer who observes ballots in person at the \localVotingCenter.

    If serial numbers are omitted, it is probably convenient to add a ballot type indicator on the paper ballot---this information is present anyway, and having it printed explicitly is convenient for processing at the \localVotingCenter.

    Call the (small) set of \protname{}\ ballots on the BB $M$. When $|M|$ is small, mixing does not have much benefit.  We could use the ballots output from the mix, but it is just as valid to skip mixing  and use the aggregate of $M$ from \emph{before} the mixing step that is published in Step~\ref{step:digitalEncryptedTally} of the digital path.

    The main idea is to make a larger batch that includes all the \protname{}\ ballots and enough ordinary ballots to constitute a reasonable anonymity set of ballots of the same type.  We tally the batch electronically, then use one of several possible methods for incorporating  it into an existing audit. Table \ref{Tab:privacysubmerge} presents an overview of the access levels granted to various stakeholders in Sub-\protname{}. Plaintext sub-tallies from \protname{} ballots are not calculated.\footnote{They can of course be inferred by inspecting the arriving ballot papers, but this is not part of the privacy attacker model for this section.}

    \begin{table}
\begin{center}
  \caption{Data accessible to various entities in Sub-\protname\ where SN is replaced with an explicit ballot type identifier.}
  \label{Tab:privacysubmerge}

  \begin{tabular}{||l||c|c|c|c||}
      \hline
      \hline
      Entity&Voter&Vote&\protname{}&Ballot\\
        &ID& &Tally&type\\

      \hline
      General public&$\checkmark$&$\times$&$\times$&$\checkmark$\\
      BMD machine&$\checkmark$&$\checkmark$&$\checkmark$&$\checkmark$\\
      Coercer&$\checkmark$&$\times$&$\times$&$\checkmark$\\
      RLA observers&$\times$&$\checkmark$&$\times$&$\checkmark$\\
      Postal worker&$\checkmark$&$\times$&$\times$&$\times$\\
      \hline\hline
  \end{tabular}
    \end{center}
 \end{table}

    \begin{enumerate}
    \item Gather a collection of local ballots $L$ from those available in the \localVotingCenter{}. They may have been cast in person, or they may be other (non-\protname) absentee ballots. $L$ should be chosen so that, when combined with the \protname{}\ ballots, the anonymity set is large enough for the tallies to be published. Probably $|L \cup M| \approx 30.$

    We assume that ballots in $L$ are already disassociated from the voter's name.

    \item  Encrypt the ballots in $L$ using the encryption scheme from \autoref{subsec:notationAndBuildingBlocks}, including validity proofs. Post them on the BB.
        \item Use the homomorphic property to compute the combined aggregate of $L \cup M$. Decrypt it and publish the ciphertext with proof of correct decryption on the BB. This gives us, on the BB:
         $$\enc(\mathbf{Tally_{L \cup M}}), \mathbf{Tally_{L \cup M}}, \DecryptZKP{\mathbf{Tally_{L \cup M}}}{\enc(\mathbf {Tally_{L \cup M}}})$$
    \end{enumerate}

    The resulting data can be incorporated into a ballot-level comparison RLA by
    adopting the one-audit approach (see \autoref{subsec:oneAudit-submerge}), or into a batch-level comparison
    RLA by simply tallying the batch manually (see below).

    \subsubsection{Batch-\protname: batch-level comparison audits}
    \label{subsec:batch}
    If the \localVotingCenter{} is already using  batch-level comparison audits, these batches can easily be
    incorporated. We can  define the batches according to some physical convenience function---the corresponding electronic records should be easy to identify and need not be together on the BB. For instance, if batches are created based on the envelope arrival date, once a validly signed envelope is received, we place its ballot in the appropriate batch box and record the batch number on the BB.

    This would make sense if the \localVotingCenter{}\ already used batch-level comparison audits for their RLAs. The next audit steps would be as follows.

\begin{enumerate}
    \item[4.] Manually tally the combined batch comprising $L$ and the \protname{}\ ballots.
        \item[5.] Compute the overall discrepancy $D$ between the electronic and manual tallies.
\end{enumerate}

    We now have the discrepancy between a plaintext electronic tally and a manual tally of paper ballots, each comprising both \protname{}\ and ordinary ballots. This should be incorporated into the RLA statistics according to the rules for
    batch-level comparison RLAs, which should automatically deal appropriately with non-arrived \protname{} ballots.

    \subsubsection{Ballot-level comparison RLA using ONEAudit} \label{subsec:oneAudit-submerge}
    If the \localVotingCenter{} is already using ballot-level comparison audits, sub-\protname{} batches can be efficiently
    incorporated using ONEAudit \cite{stark2023overstatement}. In this approach, discrepancies are computed by comparing ballots to their \emph{overstatement net equivalent} CVRs (as defined below), rather than to individual CVRs.

    \begin{definition}
        Two sets with the same number  of CVRs  are \emph{overstatement net equivalent}
        if they produce the same totals.
    \end{definition}

    In this case, if a \protname{} ballot is selected for audit, but no corresponding paper ballot has arrived, it is important to ensure the right worst-case assumption in the case that overstatements of +2 or more are possible. The next steps in the audit would be as follows.
    \begin{enumerate}
        \item[4.] Count the total number of \protname{} electronic ballots in the batch, and subtract the number of \protname{} ballots that have arrived, then explicitly add that number of ``not arrived'' zombie ballots to the catalogue of ballot papers.
        \item[5.] The RLA then samples randomly from the catalogue of ballots, which includes:
        \begin{itemize}
            \item batched \protname{} paper ballots that have arrived (including any incorporated from the \localVotingCenter),
            \item zombies representing batched \protname{} ballots that have not arrived,
            \item ordinary paper ballots from the \localVotingCenter{}, if there are any, (except the ones included in Batch-\protname).
        \end{itemize}
        It could also include any \protname{} that were not included in sub-\protname{} batches, which are sampled from the output of the mix as described in \autoref{subsec:RLA}---the different approaches can determined separately for different types of ballots.
    \end{enumerate}

    If an ordinary paper ballot from the \localVotingCenter{} (not in $L$) is sampled, it is dealt with however the \localVotingCenter{} already deals with its ballot-comparison RLAs; if a (non-batched) \protname{} ballot is sampled, it is matched according to serial number and audited as described in \autoref{subsec:RLA}.
        If a ballot in $L \cup M$ is sampled, use the discrepancy between its paper ballot and the \emph{Overstatement Net Equivalent} CVR for $L \cup M$. If the ballot has not arrived, its ``paper ballot'' will be a zombie---account for it in the RLA as described in \autoref{subsec:RLAsAndOneAudit}.

\subsubsection{Batching strategies and privacy implications}
    It may even be possible to determine in advance which paper ballots will be in which batch---for example, if so few \protname{}\ ballots are sent to one \localVotingCenter{}\ that they are all treated as a single batch. Alternatively,  batches could be assigned as the ballot papers arrive, based on their sticker.

    Note that it does not matter whether the attacker knows, or can change, which ballots are contained in the same batch---this worst-case attack assumption is already part of the attacker model of batch RLAs. As long as the attacker does not know which batches will be sampled for audit, the risk limit is met.

    Table \ref{Tab:privacy-batch} summarizes the access levels
granted to various stakeholders in Batch-\protname{}.

\begin{table}
\begin{center}
  \caption{Data accessible in Batch-\protname.}
  \label{Tab:privacy-batch}

\begin{tabular}{||l||c|c|c|c|c||}
 \hline
 \hline

  	Entity&ID&Batch&ID\&&Vote&Tally\\
     &&No.&Batch No.&&\\

 \hline
General public&$\checkmark$&$\checkmark$&$\checkmark$&$\times$&$\checkmark^{\dagger}$\\

BMD machine&$\checkmark$&$\checkmark$&$\checkmark$&$\checkmark$&$\checkmark$\\
Coercer&$\checkmark$&$\checkmark$&$\checkmark$&$\times$&$\checkmark^{\dagger}$\\
RLA observers&$\times$&$\checkmark$&$\times$&$\checkmark$&$\checkmark^{\dagger}$\\
Postal worker&$\checkmark$&$\checkmark$&$\checkmark$&$\times$&$\checkmark^{\dagger}$\\

\hline\hline
  \end{tabular}
{\small \item $^\dagger$: visible on the BB}
  \end{center}
 \end{table}

In the context of this table, we operate under the assumption that the batch number assigned to each participating voter is publicly disclosed on the BB. As indicated in the table, apart from the BMD machine, exclusive access to each ballot's vote content is granted only to the observers at the \voteCollectionCenter. Nevertheless, the observers possess solely the batch number without any additional information to establish a connection between the vote contents and the respective voter's identity. Consequently, with a sufficiently large batch size, the privacy of the voter is preserved.
\end{document}